\newif\ifarxiv
\arxivtrue
\newif\ifjournal
\journalfalse

\ifjournal
\documentclass{birkjour}
\else

\documentclass{article}
\ifarxiv
\usepackage{times}
\usepackage{amsmath,amsthm,amsbsy}
\usepackage{savesym}
\savesymbol{iint}
\savesymbol{openbox}
\usepackage{txfonts}
\restoresymbol{TXF}{iint}
\restoresymbol{TXF}{openbox}
\else
\usepackage[lf]{MinionPro}
\usepackage{amsmath,amsthm,amsbsy}
\fi
\fi

\usepackage{geometry}
\usepackage{color}
\definecolor{labelkey}{rgb}{0,0,.75}
\definecolor{MyGreen}{rgb}{0,.6,.2}
\definecolor{MyDarkBlue}{rgb}{.1,.1,.75}
\usepackage[colorlinks=true, citecolor=MyGreen, linkcolor=MyDarkBlue]{hyperref}
\usepackage{tensor}
\usepackage{xypic}
\usepackage{graphicx}

\let\PsiFoo\Theta

\newcommand*\Bell{{\ensuremath{\boldsymbol\ell}}}
\DeclareMathOperator{\sgn}{\mathrm{sgn}}
\DeclareMathOperator{\osc}{\mathrm{osc}}
\DeclareMathOperator{\Ric}{\mathrm{Ric}}
\DeclareMathOperator{\ck}{\bf L}
\renewcommand{\div}{\mathop{\rm div}\nolimits}
\DeclareMathOperator{\Lap}{\Delta}

\DeclareMathOperator{\tr}{\rm tr}

\DeclareMathOperator{\diag}{\rm diag}

\newcommand{\<}{\mkern 1mu}

\def\ip<#1,#2>{\left<#1,#2\right>}
\let\ol\overline

\newcommand{\ra}{\rightarrow}

\newcommand{\Reals}{\mathbb{R}}

\newcommand{\Ints}{\mathbb{Z}}

\newcommand{\calB}{\mathcal{B}}
\newcommand{\calC}{\mathcal{C}}

\newcommand{\calK}{\mathcal{K}}

\newcommand{\calT}{\mathcal{T}}

\def\define#1{{\bf #1}}

\ifjournal
\else
\fi

\begin{document}

\newtheorem{theorem}{Theorem}[section]
\newtheorem*{proposition-cmc-summary}{Summary of Proposition \ref{prop:cmcflat}}
\newtheorem*{proposition-noncmc-summary}{Summary of Propositions \ref{prop:noncmc} and \ref{prop:class}}
\newtheorem{proposition}[theorem]{Proposition}
\newtheorem{corollary}[theorem]{Corollary}
\newtheorem{lemma}[theorem]{Lemma}
\theoremstyle{definition}
\newtheorem{definition}[theorem]{Definition}
\numberwithin{equation}{section}

\date{April 29, 2014}

\ifjournal
\title[Conformal Parameterizations of Flat Kasner Spacetimes]
{Conformal Parameterizations of Slices of\\ Flat Kasner Spacetimes}
\else
\title{Conformal Parameterizations of Slices of Flat Kasner Spacetimes}
\fi

\author{David Maxwell}
\ifjournal
\address{%
Department of Mathematics and Statistics\\
University of Alaska Fairbanks\\
P.O. Box 756660\\
Fairbanks, AK 99775-6660}
\email{damaxwell@alaska.edu}
\fi

\ifjournal\else\maketitle\fi

\begin{abstract}
The Kasner metrics are among the simplest solutions of the vacuum Einstein equations,
and we use them here to examine the conformal method of finding
solutions of the Einstein constraint equations. After describing the 
conformal method's construction
of constant mean curvature (CMC) slices of Kasner spacetimes, we turn our attention
to non-CMC slices of the smaller family of flat Kasner spacetimes.
In this restricted setting we obtain a full description of the construction of 
certain $U^{n-1}$ symmetric slices, even in the far-from-CMC regime. Among the
conformal data sets generating these slices we find that most data sets construct
a single flat Kasner spacetime, but that there are also far-from-CMC data sets 
that construct one-parameter families of slices.
Although these non-CMC families are analogues of well-known CMC one-parameter
families, they differ in important ways.  Most significantly,
unlike the CMC case, the condition signaling the appearance of these non-CMC families
is not naturally detected from the conformal data set itself.  In light of this difficulty, 
we propose modifications of the conformal method that involve a conformally transforming
mean curvature.
\end{abstract}

\ifjournal
\subjclass{Primary 	83C05; Secondary 53Z05 } 
\keywords{conformal method, Einstein constraint equations, Kasner}
\fi

\ifjournal\maketitle\fi

\section{Introduction}

Initial data for the vacuum Cauchy problem in general relativity
consist of a Riemannian manifold $(M^n,g)$ and
symmetric (0,2)-tensor $K$ on $M$. The goal of the Cauchy problem is to 
find a Lorenzian manifold $(\Lambda^{n+1},\lambda)$  satisfying
the vacuum Einstein equation $\Ric_{\lambda} = 0$
and an embedding of $M$ into $\Lambda$ such that $g$ and $K$
are the induced metric and second-fundamental form.  Since $\lambda$
is Ricci flat, the Gauss and Codazzi equations impose the following necessary
conditions on $g$ and $K$ for the existence of a solution
of the Cauchy problem:
\begin{subequations}\label{eq:constraints}
\begin{alignat}{2}
R_{g} - |K|_{g}^2 + (\tr_g K)^2 &= 0 &\qquad&\text{\small[Hamiltonian constraint]}\label{eq:hamiltonian}\\
\div_{g} K &= d \tau&\qquad&\text{\small[momentum constraint]}\label{eq:momentum}
\end{alignat}
\end{subequations}
where $R_g$ is the scalar curvature of $g$ and $\tau = \tr_g K$ is the mean curvature.
These are known as the Einstein constraint equations, and
Choquet-Bruhat showed \cite{FouresBruhat:1952wg} that the Cauchy problem is solvable
if and only if $g$ and $K$ solve equations \eqref{eq:constraints}.

The oldest and most general method of solving the constraints is the conformal method, which was
initiated in 1944 by Lichnerowicz \cite{Lichnerowicz:1944} to construct solutions of the constraint
equations satisfying $\tau\equiv 0$.
The method was extended by York and collaborators: first in the 1970s to handle non-trivial 
mean curvatures  \cite{York:1973fl} \cite{OMurchadha:1974bf} and more recently 
as the conformal thin sandwich method in its Lagrangian \cite{YorkJr:1999jo} and Hamiltonian
\cite{Pfeiffer:2003ka} forms.  Although 
the original conformal method of the 1970s and the conformal thin sandwich method appear
to be different, it has been recently demonstrated that they are the same \cite{Maxwell:2014a}.

In this paper we work with the Hamiltonian conformal 
thin-sandwich method in the form presented in \cite{Maxwell:2014a} 
as the CTS-H method; we briefly summarize
it here and refer the reader to \cite{Maxwell:2014a} for a better-motivated treatment.
Given a Riemannian metric $g$, a transverse traceless tensor $\sigma$ (i.e.,
a symmetric (0,2)-tensor satisfying $\div_g \sigma=0$ and $\tr_g \sigma = 0$),
a mean curvature $\tau$, and a positive function $N$ representing a densitized lapse,
and we seek a solution $(\ol g, \ol K)$ of the constraint equations 
satisfying
\begin{subequations}\label{eq:CTSHparams-coords}
\begin{alignat}{2}
\ol g &= \phi^{q-2} g\\
\ol K &= \phi^{-2}\left( \sigma + \frac{1}{2N}(\ck_{g}W)\right) + \frac{\tau}{n}\ol g
\end{alignat}
\end{subequations}
where $\ck_g$ is the conformal Killing operator,
\begin{equation}
q = \frac{2n}{n-2},
\end{equation}
and where $\phi$ and $W$ are an unknown 
conformal factor and vector field respectively.
Setting
\begin{equation}
\kappa = \frac{n-1}{n}
\end{equation}
the constraint equations then
become the CTS-H equations
\ifjournal
\begin{subequations}\label{eq:ctsh}
\begin{alignat}{2}
-2\kappa q \Lap_g \phi + R_g \phi -\left|\sigma + \frac{1}{2N}\ck_g W\right|^2_g\phi^{-q-1} + \kappa \tau^2\phi^{q-1} &= 0\\
\div_g \left(\frac{1}{2N} \ck_g W\right) &= \kappa\phi^q d\tau,
\end{alignat}
\end{subequations}
\else
\begin{subequations}\label{eq:ctsh}
\begin{alignat}{2}
-2\kappa q \Lap_g \phi + R_g \phi -\left|\sigma + \frac{1}{2N}\ck_g W\right|^2_g\phi^{-q-1} + \kappa \tau^2\phi^{q-1} &= 0
&\qquad \text{\small[Lichnerowicz-York equation]}
\\
\div_g \left(\frac{1}{2N} \ck_g W\right) &= \kappa\phi^q d\tau,
&\qquad \text{\small[CTS-H momentum constraint]}
\end{alignat}
\end{subequations}
\fi
to be solved for $\phi$  and $W$. 
% Although the appearance of $N$ in the CTS-H equations appears to be 
% an additional parameter compared to the equations of the standard conformal method (which are 
% equations \eqref{eq:ctsh} with $N\equiv1/2$), it is not.  
% Working with a fixed background metric and an arbitrary
% choice of $N$ is equivalent to working with an arbitrary choice of conformal class representative $g$ and
% a fixed choice of $N\equiv1/2$, as described in \cite{Maxwell:2014a}.
If $\tau\equiv \tau_0$ for some constant $\tau_0$ (the constant mean curvature (CMC) case), 
then the choice of densitized lapse becomes irrelevant and
the solution of system \eqref{eq:ctsh} is $(\phi,W)$ where $W$ is any conformal Killing field 
(including $W\equiv 0$) and
where $\phi$ solves the CMC Lichnerowicz-York equation
\begin{equation}\label{eq:CMC-LY}
-2\kappa q \Lap_g \phi + R_g \phi -\left|\sigma\right|^2_g\phi^{-q-1} + \kappa \tau_0^2\phi^{q-1}.
\end{equation}

For simplicity we restrict our attention to the case where $M$ is compact and without boundary.
Given a conformal data set $(g,\sigma,\tau,N)$, one wants to know how many solutions $(\phi,W)$
exist for the CTS-H equations \eqref{eq:ctsh}.  Ideally there will be exactly one
(up to adding a conformal Killing field to $W$), and 
hence exactly one associated solution of the constraint equations.  
If $\tau$ is constant, this is essentially true \cite{Isenberg:1995bi}.
A conformal data set $(g, \sigma, \tau\equiv\tau_0, N)$ leads to to a unique
solution of the constraint equations except under the following circumstances:
\begin{enumerate}
\item[$\mathcal Y_+$:] $(M,g)$ is Yamabe positive and $\sigma\equiv 0$, in which case there is no solution,
\item[$\mathcal Y_-$:] $(M,g)$ is Yamabe negative and $\tau_0=0$, in which case there is no solution,
\item[$\mathcal Y_0$:] $(M,g)$ is Yamabe null and $\sigma\equiv 0$ or $\tau_0=0$ , in which case there is no solution
unless both vanish (and we then pick up a one-parameter homothety family of solutions).
\end{enumerate}
When $\tau$ is near-constant
it is also usually true that there is exactly one solution
\cite{Isenberg:1996fi}\cite{Allen:2008ef}\cite{Maxwell:2014a}, but very little is known in
the far-from-CMC setting.  
The first significant far-from-CMC result was proved in \cite{Holst:2009ce} and
extended in \cite{Maxwell:2009co}.  These papers show that given a generic Yamabe-positive
metric $g$ and an arbitrary mean curvature $\tau$ and densitized lapse $N$, if $\sigma$
is sufficiently small, then there exists at least one solution of the CTS-H 
equations, but the number of solutions is unknown.\footnote{The results of \cite{Holst:2009ce} and 
\cite{Maxwell:2009co} were proved
using the original conformal method, but using \cite{Maxwell:2014a} they imply equivalent 
results for the CTS-H method.} After forty years of study, these are the only general results known
for the conformal method, on a compact manifold, in the far-from-CMC regime. 

Given this lack of progress, it seems useful to turn to special cases to
illuminate properties that might be expected for general mean curvatures.
To this end, we consider the family of Kasner spacetimes, which
form a particularly simple class of solutions of the Einstein equations.
Within the Kasner family there is a small subfamily of flat spacetimes,
and these are the focus of our investigation of far-from-CMC solutions.

Before examining the far-from-CMC case, it is helpful to start with
CMC slices of arbitrary Kasner spacetimes. Section \ref{sec:CMC} 
describes a tidy parameterization of these slices where, in effect, the conformal parameters 
first select a member of a Kasner family and then a CMC slice within that 
spacetime. The flat Kasner spacetimes of principal interest
arise from CMC conformal data sets of the form
\begin{equation}\label{eq:cmckasnerset}
(g,\mu \sigma^\flat, \tau_0)
\end{equation}
where $g$ is a flat product metric on the torus, $\sigma^\flat$ is a particular
transverse-traceless (TT) tensor on the torus, and $\mu$ and $\tau_0$
are constants. There are restrictions on $\mu$ and $\tau_0$, however.
\begin{proposition-cmc-summary}\strut
A conformal data set \eqref{eq:cmckasnerset} generates solutions of the 
Einstein constraint equations as follows.
\begin{enumerate}
\item If $\mu$ and $\tau_0$ have the same non-zero sign, 
we obtain a CMC slice of a flat Kasner spacetime.
\item If $\mu$ and $\tau_0$ have opposite non-zero signs, we obtain a CMC slice of a member of a different category of dual-to-flat Kasner spacetimes.
\item If $\mu$ and $\tau_0$ both vanish 
we obtain 
a homothety family of CMC slices of certain non-Kasner (static-toroidal) flat spacetimes.  
\item If one of 
$\mu$ or $\tau_0$ vanishes but the other does not, then
there is no solution of the constraint equations.
\end{enumerate}
\end{proposition-cmc-summary}

Turning to non-CMC solutions of the constraint equations, 
in Section \ref{sec:non-cmc} we consider conformal data sets of the form
\begin{equation}\label{eq:noncmckasnerset}
(g,\mu\sigma^\flat, \tau, N)
\end{equation}
where $g$, $\mu$, and $\sigma^\flat$ are as in data sets \eqref{eq:cmckasnerset}, and where $N>0$
and $\tau$ are arbitrary functions of one factor of the torus.  
\ifjournal\else

\goodbreak\fi
The following two quantities play an important role 
for conformal data sets \eqref{eq:noncmckasnerset}:
\begin{equation}
\label{eq:tau-star-intro}
\tau^* = \frac{\int_{S^1} N \tau\; dx}{\int_{S^1} N\; dx}
\end{equation}
and
\begin{equation}
\tau^\circ = \frac{\int_{S^1} \tau\; dx}{\int_{S^1} 1\; dx}
\label{eq:tau-circ-intro}
\end{equation}
where $x$ is the standard coordinate on $S^1$.
\begin{proposition-noncmc-summary}
Consider a conformal data set \eqref{eq:noncmckasnerset}.
\begin{enumerate}
\item If $\mu$ and $\tau^*$ are non-zero and have the same sign,
then the conformal data set generates at least one solution
of the constraint equations.  If $\tau^\circ\neq 0$ 
it is a flat Kasner solution, otherwise it is a 
static-toroidal solution.
\item If $\mu$ and $\tau^*$ both
vanish, then the conformal data set generates a one-parameter
family of solutions.
If $\tau^\circ\neq 0$ 
the family consists of flat Kasner solutions, otherwise
the family consists of static-toroidal solutions.
\end{enumerate}
\end{proposition-noncmc-summary}

Hence Propositions \ref{prop:noncmc} and \ref{prop:class} 
establish non-CMC analogues of cases 1) and 3)
of Proposition \ref{prop:cmc}; cases 2) and 4) remain
open.

The conditions $\tau^*=0$ and $\tau^\circ=0$ are separate exceptional
cases, the first associated with one-parameter families of solutions,
and the second signaling static-toroidal solutions rather than flat Kasner
solutions. Now if $\tau\equiv\tau_0$ for some constant $\tau_0$, then
$\tau^*=\tau^\circ=\tau_0$ and the two exceptional cases
merge into a single exceptional condition: $\tau_0=0$.  So in the
CMC setting the one-parameter families always
coincide with the static-toroidal solutions.
In contrast, for non-CMC data, the one-parameter families are typically
flat Kasner solutions (with static-toroidal families appearing
only in the rare case when $\tau^\circ=0$ as well).

The most important difference between the CMC and non-CMC results, however, is that it
is comparatively difficult to tell from a general non-CMC conformal data set if 
it is associated with a one-parameter family of flat Kasner spacetimes.
In the CMC setting we specify $\tau_0$ directly, and $\tau_0=0$ is 
the exceptional case.  For non-CMC solutions, we must correspondingly determine if
$\tau^*=0$. Although $\tau^*$ appears to be directly 
specified in equation \eqref{eq:tau-star-intro} from $\tau$ and $N$,
this is something of an illusion.
Recall that the CTS-H method is conformally covariant
in the sense that if $(g,\sigma,\tau,N)$ and $(\tilde g,\tilde \sigma,\tilde \tau,\tilde N)$
are related by
\begin{equation}
\begin{aligned}
\tilde g &= \psi^{q-2} g\\
\tilde \sigma &= \psi^{-2} \sigma\\
\tilde \tau &= \tau\\
\tilde N &= \psi^q N
\end{aligned}
\end{equation}
for some conformal factor $\psi$,
then the two sets of conformal data generate the same solutions of the constraint equations.
So if $(g,\sigma,\tau,N)$ generates a one-parameter family, so does 
$(\tilde g,\tilde \sigma,\tilde \tau,\tilde N)$.  But equation \eqref{eq:tau-star-intro} defining $\tau^*$
is not conformally covariant: the value of $\tau^*$ is readily computed with
respect to only a few choices of background metric.  As discussed
in Section \ref{sec:mc-alt}, it appears that one cannot determine if 
a data set $(\tilde g, \tilde \sigma, \tilde \tau, \tilde N)$
will generate a one-parameter family of
solutions without first effectively solving the CTS-H equations in the first place.
We emphasize that the CTS-H equations are equivalent to the classic non-CMC conformal
method of \cite{York:1973fl}, so the difficulty of detecting $\tau^*=0$ cannot be remedied
by changing to that formulation: there exist non-CMC classic conformal data sets $(g,\sigma,\tau)$
that lead to one-parameter families of solutions, and the condition determining
whether or not a one-parameter family occurs involves first finding, up to some scale, a solution metric.
On the other hand, since $\tau^*$ 
is easy to compute with respect to the solution metric, this suggests alternative,
non-trivial conformal transformation laws for the mean curvature that we describe in Section 
\ref{sec:mc-alt}.

It is remarkable that the one-parameter families presented here have remained
undetected for as long as they have.  Once one knows to look for them,
they are easy to find; the harder work is showing that the families are slices
of flat Kasner spacetimes.  We note, however, that a special case of the families 
was found previously in \cite{Maxwell:2011if}, which was another case study of 
special conformal data sets on the torus, and which discovered a number 
of non-existence/non-uniqueness phenomena in the far-from-CMC regime, including
occasional one-parameter families.  
The mean curvatures in that study are functions of one factor of the torus 
of the following form:
\begin{equation}\label{eq:jump}
\tau_a(x) = a + \xi(x)
\end{equation}
where $a$ is a constant and where $\xi$ equals $1$ on half of the circle and equals $-1$
on the other half.  Typically one assumes that a mean curvature belongs to
$W^{1,p}$ with $p>n$, so the $L^\infty$ regularity of the mean curvatures in 
\cite{Maxwell:2011if} made it conceivable that the phenomena discovered there
did not extend to more traditional conformal data sets.  We find here that not
only are there smooth non-CMC conformal data sets that lead to one-parameter families
of solutions, but, as discussed in Section \ref{sec:jump},
at least some of the solutions of the constraint equations
found in \cite{Maxwell:2011if} for non-smooth conformal data
generate smooth ambient spacetimes.  

\section{Preliminaries}\label{sec:prelims}

Let $\Reals^{1,1}$ be $\Reals^2$ equipped with the Lorenzian metric $\lambda=\diag(-1,1)$.  
We say a vector $V=(t,x)$ is future pointing if $t>0$ and past pointing if $t<0$.
Similarly, $V$ is rightward or leftward pointing if $x>0$ or $x<0$ respectively.
The open triangles of future and past-pointing timelike vectors are denoted by $I_+$ and $I_-$.

Given $\psi\in\Reals$ we define the future-pointing unit timelike vector
\begin{equation}
T(\psi) = (\cosh(\psi),\sinh(\psi))
\end{equation}
and the rightward-pointing unit spacelike vector
\begin{equation}
X(\psi) = (\sinh(\psi), \cosh(\psi) ).
\end{equation}

Given a hyperbolic angle $\Psi\in\Reals$ we define $B_\Psi$ to be the boost on $\Reals^{1,1}$
that fixes the origin and takes $T(0)$ to $T(\Psi)$. This isometry is a linear map, and its matrix with
respect to the standard basis $\{T(0),X(0)\}$ is 
\begin{equation}
\begin{pmatrix} \cosh(\Psi) & \sinh(\Psi) \\
\sinh(\Psi) & \cosh(\Psi) \end{pmatrix}.
\end{equation}
Elementary arguments show that for any $\psi\in\Reals$, $B_\Psi(T(\psi)) = T(\psi+\Psi)$
and $B_\Psi(X(\psi))=X(\psi+\Psi)$.

We take $S^1$ to be the circle of radius 1 in $\Reals^2$, and its \define{standard metric}
is the metric inherited from this embedding.
The covering map $s\mapsto (\cos(2\pi s),\sin(2\pi s))$ provides what we will call
\textbf{unit coordinates} on $S^1$, and 
the vector field $\partial_s$ determines the positive orientation on $S^1$.
Given a length $\ell>0$, we define the metric 
\begin{equation}
g_\ell = \ell^2 ds^2
\end{equation}
on $S^1$, so the standard metric is $g_{2\pi}$.
We write $S^1_\ell$ for $S^1$ equipped with $g_\ell$.

\define{Unit coordinates} on tori $T^n$ are products of unit coordinates on $S^1$.
Given a vector $\Bell=(\ell_1,\ldots,\ell_n)$ of lengths we define the
product metric on $T^n$
\begin{equation}
g_\Bell = \sum_{k=1}^n \ell_k^2 (ds^k)^2
\end{equation}
in unit coordinates, and $T^n$ equipped with $g_\Bell$ is $T^n_\Bell$.
The volume element on $T^n_\Bell$ is $dV_\Bell$.

We will mostly assume that the objects under discussion are smooth.
At times, however, we will discuss curves in Lorenzian surfaces
with $L^\infty$ curvatures, and we clarify here what this means.
Let $\gamma$ be a curve from an interval $I$ (with coordinate $s$) 
to a smooth, time-oriented Lorenzian surface $\Lambda^{1,1}$.
We say the curve is $W^{2,p}$ for some $p\ge 1$ if in any smooth coordinate chart the components
of $\gamma$ belong to $W^{2,p}(I)$. Similar considerations hold for curves
defined on $S^1$ rather than $I$; we replace $s$ with the unit coordinate on $S^1$.

Now assume that $\gamma\in W^{2,p}$ with $p>1$, so
$\gamma$ is $C^{1,\alpha}$ with $(1/\alpha) = 1 - (1/p)$.
In particular the tangent vector to such a curve is well-defined.
Suppose $\gamma'$ is spacelike everywhere.  
We can then pull back the metric $\lambda$ on $\Lambda^{1,1}$ to obtain
a $W^{1,p}$ Riemannian metric $g$ on $I$ or $S^1$ via
\begin{equation}
g(\partial_s, \partial_s) = \ip<\gamma'(s),\gamma'(s)>_\lambda.
\end{equation}
For each $s\in I$ (or $S^1$) we let $n(s)$ be the future-pointing unit timelike vector
at $\gamma(s)$ 
that is orthogonal to $\gamma'(s)$. This is a $W^{1,p}$ vector field along
the $W^{2,p}$ curve $\gamma$ and we define the induced second-fundamental form $K$ by
\begin{equation}
K(\partial_s,\partial_s) = -\ip< n(s), \nabla_{\gamma'} \gamma'>,
\end{equation}
where $\nabla_{\gamma'} \gamma'$ is the $L^p$ valued vector field defined along $\gamma$
via the usual formula in any local coordinates.
The curvature of $\gamma$ is the function $\tau\in L^p$ such that
\begin{equation}
K(\partial_s,\partial_s) = \tau(s) g(\partial_s,\partial_s).
\end{equation}
If $\Lambda^{1,1}=\Reals^{1,1}$ and $\gamma(s)=(t(s),x(s))$,
\begin{equation}\label{eq:coordcurvature}
\tau(s) =  
\frac{t''(s)x'-x''(s)t'(s)}{{(x'(s)^2-t'(s)^2)}^{3/2}}.
\end{equation}

\section{Kasner Spacetimes}\label{sec:Kasner}

The \define{expanding Kasner metrics} 
on $\Reals_+ \times \Reals^n$, in coordinates $(t,x^1,\ldots,x^n)$, have the form
\begin{equation}\label{eq:Kasner}
\lambda = -dt^2 + \sum_{k=1}^n |t|^{2 a_k} (dx^k)^2
\end{equation}
where the exponents $a_k$ satisfy the Kasner conditions
\begin{equation}\label{eq:KasnerConds}
\begin{aligned}
\sum_{i=k}^n a_k &= 1\\
\sum_{i=k}^n (a_k)^2 &= 1
\end{aligned}
\end{equation}
which ensure the metric is Ricci flat (and therefore a solution of the vacuum 
Einstein equations) \cite{Kasner:1921iy}.  
We also have the \define{contracting Kasner metrics} given by the metric \eqref{eq:Kasner}
on $\Reals_- \times \Reals^n$ rather than $\Reals_+\times \Reals^n$, which are isometric
to the expanding Kasner metrics, and differ only in their time orientations: we take $\partial_t$
to be positively oriented in both cases.

We will be concerned with spacetimes with compact Cauchy surfaces,
and obtain these from the Kasner metrics by taking a quotient by a group of isometries generated
by translations of length $\ell_k$ in the $x^k$ direction, one translation for each $k$.  The resulting
metrics on $\Reals_\pm \times T^n$ then have the form
\begin{equation}\label{eq:KasnerCpct}
-dt^2 + \sum_{k=1}^n |t|^{2 a_k} (\ell_k)^2 (ds^k)^2
\end{equation}
in unit coordinates on the torus.  Henceforth an \define{expanding} or \define{contracting Kasner spacetime} 
will refer to  $\Reals_+ \times T^n$ or  $\Reals_- \times T^n$ respectively
equipped with a  Kasner metric \eqref{eq:KasnerCpct}.

Each slice of constant $t$ in a Kasner spacetime is a flat Riemannian torus with metric
\begin{equation}\label{eq:kasnerg}
g = \sum_{k=1}^n  |t|^{2a_k}(\ell_k)^2(ds^k)^2
\end{equation}
and with second fundamental form
\begin{equation}\label{eq:kasnerK}
K=\sum_{k=1}^n  \frac{a_k}{t} |t|^{2a_k}(\ell_k)^2(ds^k)^2.
\end{equation}
Since $\sum_k a_k=1$, the slice of constant $t$ has constant mean curvature 
\begin{equation}\label{eq:KasnerCMC}
\tau = 1/t.
\end{equation}

Although the Kasner conditions rule out the possibility that all exponents are
equal, there are two exceptional cases where all exponents but one are equal.
If $a_1=1$ and the remaining exponents vanish we
obtain the \textbf{flat Kasner} spacetimes.  These metrics have
vanishing curvature, not just vanishing Ricci curvature, 
and the metric decomposes as a product
on $(\Reals_\pm\times S^1)\times T^{n-1}$.  The metric on $\Reals_\pm\times S^1$ is 
\begin{equation}\label{eq:FlatKasner}
-dt^2 + t^{2} (\ell_1)^2 (ds^1)^2 
\end{equation}
and $T^{n-1}$ has a flat product metric $g_{\hat\Bell}$ for some $\hat\Bell=(\ell_2,
\ldots, \ell_n)$.
If instead $a_1 = 2/q = (2/n)-1$ and each remaining $a_k = 2/n$, then
the Kasner metric has the form
\begin{equation}
-dt^2 + |t|^{(4/n)-2} (\ell_1)^2 (ds^1)^2 + |t|^{4/n} g_{\hat \Bell}.
\end{equation}
\nobreak
We will call these \textbf{dual-to-flat Kasner} spacetimes.  
\goodbreak

The metric \eqref{eq:FlatKasner} associated with flat
Kasner spacetimes can be obtained from the standard metric on $\Reals^{1,1}$
by taking a quotient by a group of boosts, and we sketch
this construction next.

\begin{definition}\label{def:KasnerSurface}
Let $\Psi>0$, and let $\calB_\Psi$ be the group of isometries of $\Reals^{1,1}$ generated
by the boost $B_\Psi$.  Then $\calB_\Psi$ acts smoothly and properly discontinuously on $I_+$,
and we define the \define{Kasner surface with aperture $\Psi$} by
$\calK_\Psi = I_+/\calB_\Psi$. It is a smooth manifold, and since
the quotient is by a group of time- and space-orientation preserving isometries,
$\calK_\Psi$ inherits a time- and space-oriented Lorenzian metric such that the 
projection from $I_+$ is a local (time- and space-orientation preserving) isometry.

The \define{Kasner surfaces with negative aperture} $\Psi<0$ are defined similarly, except
that $\calB_{\Psi}$ acts on $I_-$ rather than $I_+$.
\end{definition}

\begin{figure}
\ifjournal
\includegraphics{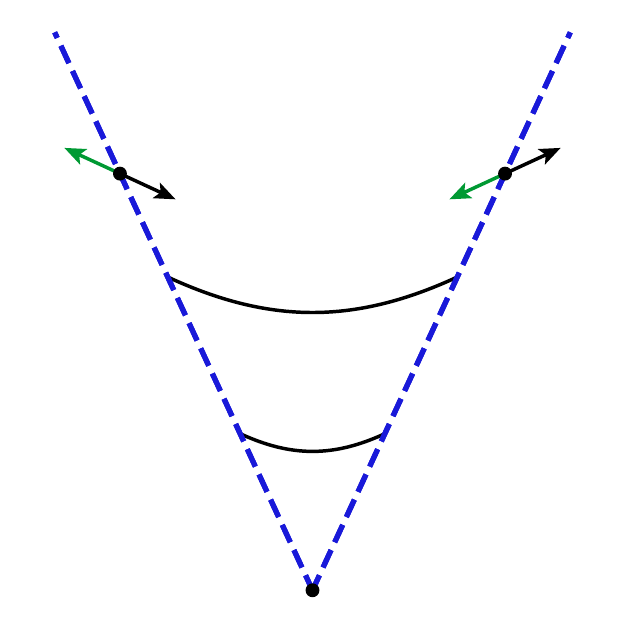}
\else
\hfil\includegraphics{boost}
\fi
\caption{\label{fig:boost}A Kasner surface with positive aperture as a quotient of Minkowski space. 
Blue dashed lines on the left and
right are identified.  The two green arrows represent the same tangent vector, as do the
two black arrows. The
solid black curves are smooth images of $S^1$ with constant curvature,
and are slices of constant $t$ for the metric \eqref{eq:Kasner2}.}
\end{figure}

Supposing $\Psi>0$, we define a covering map from $R_+\times \Reals$ to $\mathcal K_\Psi$
by
\begin{equation}
(t,x) \mapsto \Pi(t T(x)),
\end{equation}
where $\Pi$ is the projection from $I_+$ to $\calK_{\Psi}$.  The metric on 
$\calK_{\Psi}$ then pulls back to 
\begin{equation}\label{eq:Kasner2}
-dt^2 + t^2 dx^2,
\end{equation}
and taking the quotient of this spacetime by the group of isometries generated by the translation  $x\mapsto x+\Psi$ 
then determines a spacetime isometric to $\calK_\Psi$, with a metric on $\Reals_+\times S^1$
\begin{equation}\label{eq:FlatKasner2}
-dt^2 + \Psi^2 t^2 ds^2.
\end{equation}
An analogous construction is possible when $\Psi<0$, in which case
we pick up the metric \eqref{eq:FlatKasner2} on $\Reals_-\times S^1$ instead of 
$\Reals_+\times S^1$.
Comparing the metrics \eqref{eq:FlatKasner2} and \eqref{eq:FlatKasner} we see that
a flat Kasner spacetime is isometric to a product
\begin{equation}
\mathcal K_\Psi \times T^{n-1}_{\hat\Bell}
\end{equation}
for some $\Psi\neq 0$ and positive lengths $\hat \Bell=(\ell_2,\ldots,\ell_n)$.

The case $\Psi=0$ was excluded in Definition \ref{def:KasnerSurface}; in its 
place we have the following.
\begin{definition}\label{def:cylindrical}
Let $\ell>0$ and let $\calT_\ell$ be the group of isometries of $\Reals^{1,1}$ generated
by translation by $(0,\ell)$.  Then 
$\calC_\ell = \Reals^{1,1}/\calT_\ell$ is a \define{flat Lorenzian cylinder}.
A \define{static toroidal spacetime} is  $\Reals\times T^n$ with the metric
\begin{equation}
-dt^2 + g_{\Bell}
\end{equation}
for some $\Bell=(\ell_1,\ldots,\ell_n)$.  Equivalently 
it is
$\calC_{\ell_1}\times T^{n-1}_{\hat\Bell}$ 
where $\hat\Bell=(\ell_2,\ldots,\ell_n)$.
\end{definition}

\section{Conformal Constructions of CMC Slices of Kasner Spacetimes}\label{sec:CMC}

We wish to describe the conformal data sets that lead to the 
previously described CMC slices of Kasner spacetimes.
Since we know in advance that the solution metric is a flat product metric, 
conformal covariance allows us to focus on conformal data sets of the form
$(g_\Bell, \sigma, \tau_0, N)$ for some vector $\Bell$ of lengths.
Since the slices are CMC, we know that $\tau\equiv \tau_0$ for some
constant $\tau_0$, and that the densitized lapse $N$ is irrelevant. So
our task is to identify the associated transverse-traceless tensors $\sigma$.

Let $\mathbf a=(a_1,\ldots,a_n)$ be a vector of Kasner exponents.  For reasons that
will be clear shortly, the vector
\begin{equation}\label{eq:b}
\mathbf b = \mathbf a - (1/n,\ldots,1/n)
\end{equation}
plays a prominent role in conformal constructions of CMC slices of Kasner spacetimes,
and we will call it a vector of \textbf{offset Kasner exponents}.  A routine computation
shows that if $\mathbf a$ and $\mathbf b$ are related by equation \eqref{eq:b}, then
$\mathbf a$ satisfies the Kasner conditions \eqref{eq:KasnerConds} if and only if $\mathbf b$
satisfies the \textbf{offset Kasner conditions}
\begin{subequations}\label{eq:offKasner}
\begin{alignat}{2}
\sum_{k=1}^n b_k &= 0 \label{eq:offKasner1} \\
\sum_{k=1}^n (b_k)^2 &= 1 - \frac{1}{n} = \kappa.\label{eq:offKasner2}
\end{alignat}
\end{subequations}

Given vectors $\Bell$ and $\mathbf b$ of lengths and offset Kasner exponents
we define
\begin{equation}
\sigma_{\mathbf b, \Bell} = \sum_{k=1}^n b_k\, (\ell_k)^2\, (ds^k)^2
\end{equation}
in unit coordinates on the torus.  Since each $b_k$ is constant it follows that
\begin{equation}
\div_{g_{\Bell}}\, \sigma_{\mathbf b, {\Bell}} = 0,
\end{equation}
and from the first offset Kasner condition of \eqref{eq:offKasner1} we find
\begin{equation}
\tr_{g_{\Bell}} \sigma_{\mathbf b, {\Bell}} = 0
\end{equation}
as well. Hence $\sigma_{\mathbf b, \Bell}$ is transverse-traceless with respect 
to $g_{\Bell}$.  A simple computation shows that given a vector $\Bell$
of lengths, any tensor of the form 
$\sum_{k=1}^n a_k (ds^k)^2$ can be written uniquely in the form $\mu \sigma_{\mathbf b, {\Bell}}$
where $\mathbf b$ is a vector of offset Kasner exponents $\mu$ is a constant.

Now consider a conformal data set $(g_{\Bell}, \mu\sigma_{\mathbf b, {\Bell}}, \tau_0)$
for some nonzero constants $\mu$ and $\tau_0$.  
From the second offset Kasner condition
\eqref{eq:offKasner2} we see that
\begin{equation}\label{eq:normsigb}
|\sigma_{\mathbf b, {\Bell}}|_{g_{\Bell}}^2 = \kappa.
\end{equation}
and therefore the Lichnerowicz-York equation \eqref{eq:CMC-LY}
for $(g_{\Bell}, \mu\sigma_{\mathbf b, {\Bell}}, \tau_0)$
reads
\begin{equation}\label{eq:lichcmc}
-2\kappa q\Lap\phi -\mu^2 \kappa \phi^{q-1} + \kappa \tau_0^2 \phi^{q-1} = 0.
\end{equation}
Since $\mu$ and $\tau_0$ are both non-zero,
the unique positive solution \cite{Isenberg:1995bi} of equation \eqref{eq:lichcmc}
is the constant solution $\phi=c$ where
\begin{equation}\label{eq:d}
c = \left|\frac{\mu} {\tau_0}\right|^{\frac 1 q}.
\end{equation}
Therefore the pair $(\ol g, \ol K)$  given by
\begin{equation}\label{eq:cmcsol}
\begin{aligned}
\ol g &= c^{q-2} g_{\Bell}\\
\ol K &= c^{-2} \sigma_{\mathbf b, {\Bell}} + \frac{\tau_0}{n}c^{q-2} g_{\Bell}
\end{aligned}
\end{equation}
is a solution of the constraint equations, and we now show that it 
generates a Kasner spacetime.

Let 
\begin{equation}\label{eq:s}
s = c^{-q}\frac{\mu}{\tau_0}
\end{equation}
and observe from equation \eqref{eq:d} that $s=1$ if $\mu$ and $\tau_0$ have the same sign, 
and that $s=-1$ otherwise.  We then define
\begin{equation}\label{eq:a}
a_k = s b_k + \frac{1}{n}.
\end{equation}
Since $s \mathbf b$ satisfies the offset Kasner conditions \eqref{eq:offKasner}, we see that
$(a_1,\ldots, a_k)$ is a vector of Kasner exponents.  Let 
\begin{equation}\label{eq:tau2t}
t_0 = \frac{1}{\tau_0},
\end{equation}
and for each $k$ let
\begin{equation}\label{eq:barell}
\ol \ell_k  = |t_0|^{-a_k} c^{q/2-1} \ell_k.
\end{equation}
From equation \eqref{eq:barell} we obtain
\begin{equation}\label{eq:confkasnerg}
\begin{aligned}
\ol g &= \sum_{k=1}^n c^{q-2} (\ell_k)^2 (ds^k)^2\\
&= \sum_{k=1}^n |t_0|^{2a_k} (\ol \ell_k)^2 (ds^k)^2.
\end{aligned}
\end{equation}
Moreover, equations \eqref{eq:cmcsol}, \eqref{eq:s}, \eqref{eq:tau2t}, \eqref{eq:a} and \eqref{eq:barell} imply
\begin{equation}\label{eq:confkasnerK}
\begin{aligned}
\ol K &= \sum_{k=1}^n \left[c^{-2} \mu b_k +\frac{\tau_0}{n}c^{q-2}\right] (\ell_k)^2 (ds^k)^2 \\
&= \sum_{k=1}^n \tau_0 \left[ \frac{\mu c^{-q}}{\tau_0} b_k +\frac{1}{n}\right] c^{q-2}(\ell_k)^2 (ds^k)^2 \\
&= \sum_{k=1}^n \frac{1}{t} \left[ s b_k +\frac{1}{n}\right] |t_0|^{2a_k} (\ol\ell_k)^2 (ds^k)^2\\
&= \sum_{k=1}^n \frac{a_k}{t} |t_0|^{2a_k} (\ol\ell_k)^2 (ds^k)^2.
\end{aligned}
\end{equation}
Comparing equations \eqref{eq:confkasnerg} and \eqref{eq:confkasnerK} with equations 
\eqref{eq:kasnerg} and \eqref{eq:kasnerK} we see that $(\ol g, \ol K)$ is the Cauchy
data for the time $t=t_0=1/\tau_0$ slice of the Kasner spacetime with Kasner exponents
$\mathbf a$ and with metric $g_{\ol \Bell}$ at time $t=\sgn t_0 = t_0/|t_0|$.  We have therefore
demonstrated the following.

\begin{proposition}\label{prop:cmc}
Let $\Bell=(\ell_1,\ldots,\ell_n)$ be a vector of positive numbers,
let $\mathbf b=(b_1,\ldots, b_n)$ be a vector of offset Kasner exponents,
and let $\mu$ and $\tau_0$ be non-zero constants.
The conformal data set $(g_{\Bell}, \mu\> \sigma_{\mathbf b, \Bell}, \tau_0)$
on the torus $T^n$ determines a time 
\begin{equation}\label{eq:t0}
t=\frac{1}{\tau_0}\equiv t_0
\end{equation}
slice of a Kasner spacetime.  It is expanding if $\tau_0>0$ and contracting
if $\tau_0<0$.
The Kasner exponents for this spacetime
are
\begin{equation}\label{eq:a2}
\mathbf a = s \mathbf b + \left(\frac{1}{n},\ldots,\frac{1}{n}\right)
\end{equation}
where $s=1$ if $\mu$ and $\tau_0$ have the same sign and
$s=-1$ otherwise.  The time $t=t_0$ slice of the spacetime
has metric $c^{q-2} g_\Bell $
where
\begin{equation}
c = \left|\frac{\mu}{\tau_0}\right|^{\frac{1}{q}}.
\end{equation}
The time $t=\sgn{t_0}$ slice of this spacetime
has metric $g_{\ol\Bell}$ where
\begin{equation}\label{eq:ell2}
\ol \ell_k = |t_0|^{-a_k} c^{q/2-1} \ell_k.
\end{equation}
\end{proposition}

Observe that every constant time slice of a Kasner spacetime arises from the construction
of Proposition \ref{prop:cmc}.  Indeed, consider the time $t=t_0$ slice of
a Kasner spacetime with exponents $\mathbf a$, and let $g_{\hat \Bell}$ be its
time $t=\sgn(t_0)$ metric.
The time $t=t_0$ of the slice determines the mean curvature $\tau_0$ via equation \eqref{eq:t0}
and the Kasner exponents $\mathbf a$ of the spacetime determine offset Kasner exponents
$\mathbf b$ via equation \eqref{eq:a2} with $s=\sgn(\tau_0)$.  Setting $\mu = |\tau_0|$,
equation \eqref{eq:ell2} determines $\Bell$ via
\begin{equation}
\ell_k = t_0^{a_k} \ol \ell_k
\end{equation}
for each $k$, and Proposition \ref{prop:cmc} then shows that the slice is generated by 
$(g_{\Bell}, |\tau_0|\sigma_{\mathbf b, \Bell}, \tau_0)$.
By conformal covariance the slice is also generated by all conformally related data sets as well.

Proposition \ref{prop:cmc} excludes the case $\mu=0$ or $\tau_0=0$, which we turn our attention
to now.
\begin{proposition}\label{prop:zero}
Let $\Bell=(\ell_1,\ldots,\ell_n)$ be a vector of positive numbers,
let $\mathbf b=(b_1,\ldots, b_n)$ be a vector of offset Kasner exponents,
and let $\mu$ and $\tau_0$ be real numbers with at least one equal to zero.
The conformal data set $(g_{\Bell}, \mu \sigma_{\mathbf b, \Bell}, \tau_0)$
on the torus $T^n$ has no associated solution of the constraint equations
unless both $\mu=0$  and $\tau_0=0$, in which case it generates
the family of static toroidal spacetimes with metrics
\begin{equation}\label{eq:statictoroidal}
-dt^2 + \sum_{k=1}^n (r\ell_k)^2 (ds^k)^2.
\end{equation}
for every $r>0$.
\end{proposition}
\begin{proof}
The Lichnerowicz-York equation \eqref{eq:CMC-LY} for $(g_{\Bell}, \mu \sigma_{\mathbf b, \Bell}, \tau_0)$ reads
\begin{equation}\label{eq:lich2}
-2\kappa q\Lap\phi -\mu^2 \kappa \phi^{q-1} + \kappa \tau_0^2 \phi^{q-1} = 0
\end{equation}
regardless of whether $\mu$ or $\tau_0$ vanish.  Integrating this equation over the torus we find
\begin{equation}
\tau_0^2 \int_{T^n} \phi^{q-1} dV_{\mathbf \ell} = \mu^2 \int_{T^n} \phi^{-q-1} dV_{\mathbf \ell}
\end{equation}
and hence if one of $\tau_0$ or $\mu$ vanish, there is no solution unless the other vanishes as well.
If both vanish, equation \eqref{eq:lich2} is solved by $\phi=c$ for every positive constant $c$ and
we have the associated Cauchy data
\begin{equation}
\begin{aligned}
\ol g &= c^{q-2} g_\Bell \\
\ol K &= 0.
\end{aligned}
\end{equation}
Letting $r=c^{q/2-1}$, this Cauchy data evidently generates the constant $t$ 
slices of the static toroidal spacetimes with metrics \eqref{eq:statictoroidal}.
\end{proof}

Note that the Kasner exponents $\mathbf a$ are not determined uniquely by 
the choice of offset exponents $\mathbf b$ in Proposition \ref{prop:cmc}.  We see instead that
the exponents are one of
\begin{equation}\label{eq:dual}
\mathbf{a_\pm} = \pm\mathbf b + (1/n,\ldots,1/n)
\end{equation}
with the choice of sign determined by a combination of the signs of $\mu$ and $\tau_0$.
We say that the pair of Kasner exponent vectors in equation \eqref{eq:dual} are \textbf{dual}
to each other.   The flat Kasner metrics correspond, e.g., to the exponents
\begin{equation}
\mathbf a_\flat = (1, 0, \ldots, 0)
\end{equation}
and the associated offset Kasner exponents are
\begin{equation}
\begin{aligned}
\mathbf b_\flat &= (1, 0, \ldots, 0) - (1/n,\ldots,1/n)\\
& =  (\kappa, -1/n,\ldots, -1/n).
\end{aligned}
\end{equation}
Hence the dual exponents to $\mathbf a_\flat$ are
\begin{equation}
\begin{aligned}
\mathbf a_{\flat}' &= -\mathbf b_\flat + \left(\frac 1 n,\ldots,\frac 1 n\right) \\
&= \left(\frac{2-n}{n},\frac 2 n,\ldots, \frac 2 n\right)\\
&= \left(-\frac 2 q,\frac 2 n,\ldots, \frac 2 n\right),
\end{aligned}
\end{equation}
and correspond to the dual-to-flat Kasner spacetimes defined in 
Section \ref{sec:Kasner}.

We now define
\begin{equation}\label{eq:sigflat}
\sigma_\Bell^\flat = \sigma_{\mathbf b^\flat,\Bell}.
\end{equation}
From Propositions \ref{prop:cmc} and \ref{prop:zero} we see that 
conformal data set $(g_\Bell, \mu \sigma_\Bell^\flat, \tau_0)$ determines
either a flat Kasner spacetime, a dual-to-flat Kasner spacetime, a static toroidal spacetime,
or nothing depending on the values of $\mu$ and $\tau_0$.  Specifically, we have the following.

\begin{proposition}\label{prop:cmcflat}
Consider a CMC conformal data set $(g_\Bell, \mu\< \sigma^\flat_\Bell, \tau_0)$
on $T^n$ where $\Bell$ is a vector of lengths and
$\mu$ and $\tau_0$ are constants. Let $V_\Bell$ be the volume of $(T^n,g_\Bell)$; i.e.,
$V_\Bell = \ell_1\cdots\ell_n$.
\begin{enumerate}
\item If $\tau_0$ and $\mu$ are nonzero and have the same sign, then the conformal data set generates a
CMC slice of a flat Kasner spacetime with volume equal to $|\mu/\tau_0|V_\Bell$.
It is expanding if $\tau_0>0$ and contracting if $\tau_0<0$.
\item If $\tau_0$ and $\mu$ are nonzero and have opposite signs, then the conformal data set generates a 
CMC slice of a dual-to-flat Kasner spacetime
with volume equal to $|\mu/\tau_0|V_\Bell$.
It is expanding if $\tau_0>0$ and contracting if $\tau_0<0$.  
\item If $\tau_0=0$ and $\mu=0$ then the conformal data set generates a homothety family of CMC slices
of static toroidal spacetimes.
\item If one of $\tau_0$ or $\mu$ vanishes and the other does not, then the conformal data set does not generate a 
solution of the Einstein constraint equations.
\end{enumerate}
\end{proposition}
\begin{proof}
Suppose $\mu$ and $\tau_0$ are nonzero and have the same sign.  Proposition \ref{prop:cmc} implies
that a conformal data set $(g_\Bell, \mu\< \sigma^\flat_\Bell, \tau_0)$ determines a Kasner spacetime
with exponents
\begin{equation}
\mathbf a = \mathbf b_\flat + \left(\frac1 n, \ldots, \frac1 n\right) = (1, 0, \dots, 0),
\end{equation}
i.e., a flat Kasner spacetime, that is expanding or contracting depending on the sign of $\tau_0$.
The solution metric is $c^{q-2} g_\Bell$ where
\begin{equation}
c^q = \left| \frac{\mu}{\tau_0}\right|.
\end{equation} 
Since
\begin{equation}
q-2 = \frac{4}{n-2} = 2\frac{q}{n}
\end{equation}
we can rewrite the solution metric as
\begin{equation}
c^{2\frac{q}{n}} g_\Bell
\end{equation}
which has volume
\begin{equation}
c^q V_\Bell = \left| \frac{\mu}{\tau_0}\right|V_\Bell .
\end{equation}
This establishes the claim in case 1 of the proposition.

Turning to case 2, if $\mu$ and $\tau_0$ are nonzero but have opposite signs, then Proposition \ref{prop:cmc} implies
that the conformal data set generates a Kasner spacetime with exponents
\begin{equation}
\mathbf a = -\mathbf b_\flat + (1/n,\ldots,1/n) = {\mathbf a}_\flat',
\end{equation}
a dual-to-flat spacetime.  The remainder of case 2 is proved as in case 1, and cases 3 and 4 are immediate consequences
of Proposition \ref{prop:zero}.
\end{proof}

In Proposition \ref{prop:cmcflat}, if we fix the lengths $\Bell$, then there are two parameters to 
adjust: $\mu$, which controls the size of the TT tensor, and $\tau_0$, which controls the mean
curvature.  Figure \ref{fig:cmcflat} shows how $\mu$ and $\tau_0$ then determine
the associated spacetime.
The ratio $|\mu/\tau_0|$ is constant
on every line through the origin in Figure \ref{fig:cmcflat} , and therefore so
is the volume of the solution metric.  A graph of the volume of the solution as a function of the parameters
$\mu$ and $\tau_0$ is therefore a helicoid-like ruled surface with volume tending to 
infinity as $\tau_0\ra 0$ and tending to zero as $\mu$ (i.e., the size of the TT tensor)
goes to zero.  The homothety family
at $\mu=0$ and $\tau_0=0$ is the axis of the `helicoid'.

\begin{figure}
\ifjournal\else\hfil\fi
\includegraphics{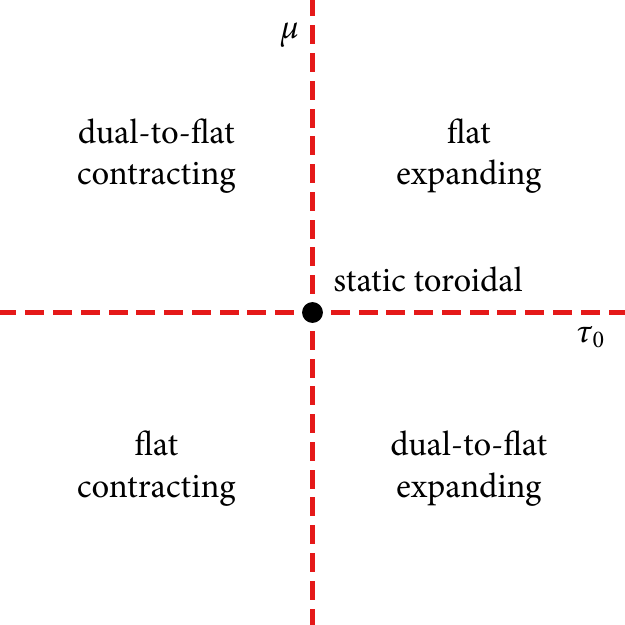}
\caption{\label{fig:cmcflat} CMC conformal method construction of slices of flat Kasner, dual-to-flat Kasner, and static
toroidal spacetimes.  Dashed red lines imply no solution, and the dot at the origin corresponds to a 
homothety family of static toroidal solutions.}
\end{figure}

\section{Spacelike Curves in Kasner Surfaces}\label{sec:curves}

We wish to generalize Proposition \ref{prop:cmcflat} to include mean curvatures
that are a function of the first factor of the torus.  To this end, we
start with the related problem of specifying a curve in a Kasner surface
in terms of its curvature.  That is, given a curvature function $\tau$ on 
$S^1$, we wish to find a Kasner surface $\mathcal K_\Psi$,
for some aperture $\Psi$, and a constant speed (with respect to the standard metric)
curve $\gamma:S^1\ra \mathcal K_\Psi$
such that $\gamma$ has the specified curvature.  The following 
proposition shows that this problem is solvable, and we can additionally
impose the length $\ell$ of the embedded curve.

\begin{proposition}\label{prop:curve}
Suppose $\tau$ is a smooth function on $S^1$,
and let $\ell>0$.
Define
\begin{equation}\label{eq:Theta}
\tau^\circ= \int_{S^1} \tau\; ds
\end{equation}
where $s$ is the unit coordinate on $S^1$ (so $\int_{S^1}1\;ds = 1$).
If $\tau^\circ\neq 0$ we construct an embedding of $S^1$ into
a Kasner surface $\calK_\Psi$ and if $\tau^\circ = 0$ we 
construct an embedding of $S^1$ into a flat Lorenzian cylinder
$\calC_L$ as described next.

Suppose $\tau^\circ\neq 0$, and let $\Psi=\ell \tau^\circ$.
Then there exists a smooth embedding $\iota:S^1\ra\calK_{\Psi}$ 
such that $\iota(S^1)$
is a spacelike curve
with induced metric $g$ and second fundamental form $K$
satisfying $\iota^* g = g_\ell$ and $\iota^* K=\tau \; g_\ell$.

If $\tau^\circ=0$, then
there exists $L>0$ and a smooth embedding $\iota:S^1\ra\calC_{L}$
such that $\iota(S^1)$
is a spacelike curve
with induced metric $g$ and second fundamental form $K$
satisfying $\iota^* g = g_\ell$ and $\iota^* K=\tau\; g_\ell$.
\end{proposition}
\begin{proof}
Lift $\tau$ to a 1-periodic function on $\Reals$ and
define for any $s\in\Reals$
\begin{equation}
\theta(s) = \int_0^s \tau(y)\; dy.
\end{equation}
Define 
\begin{equation}
\gamma_1(s) = \int_0^s \ell X(\ell\,\theta(y))\; dy
\end{equation}
where $X$ is the spacelike unit-vector function defined
in Section \ref{sec:prelims}.
So $\gamma_1$ is a smooth curve and $\gamma_1'(s) = \ell X(\ell \theta(s) )$
for all $s$.

\begin{figure}
\ifjournal\includegraphics{curveconstruction_j}\else
\hskip -0.8cm\includegraphics{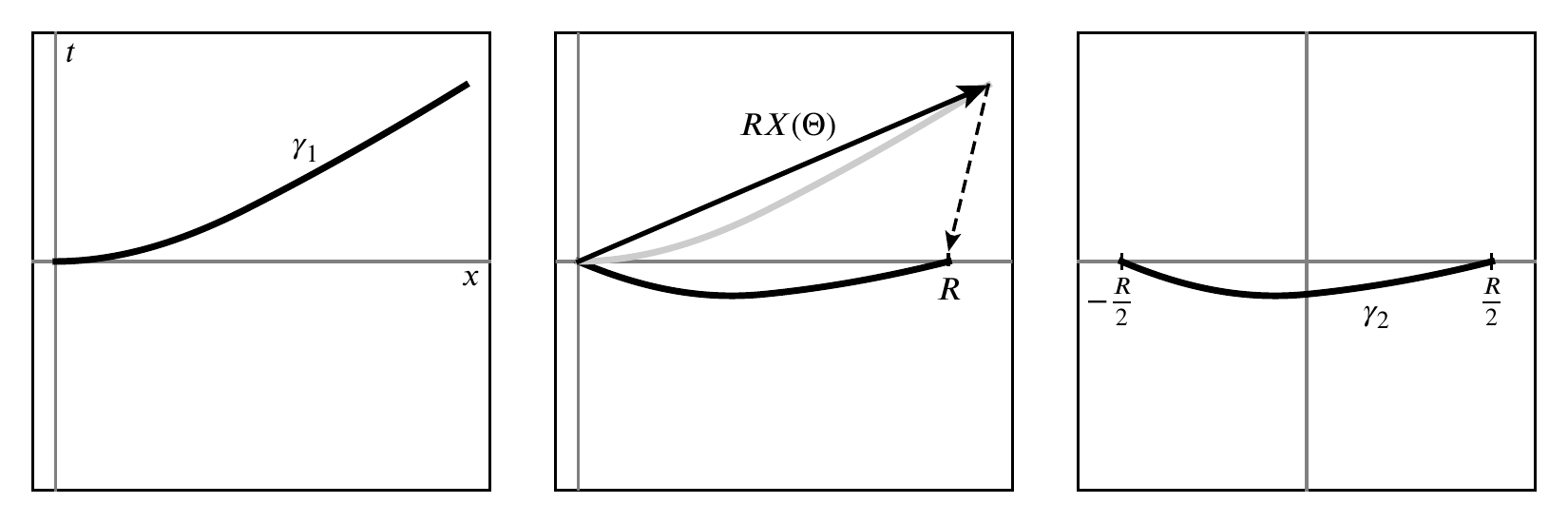}\fi

\caption{\label{fig:curve1}Transformation from $\gamma_1$ to $\gamma_2$: first boost, then translate
to obtain a curve with endpoints located at antipodal points on the $x$ axis.}
\end{figure}

Since $\gamma_1(1)$ is an integral of rightward-pointing spacelike
vectors, it is itself rightward pointing and spacelike.  So there
is $R>0$ and $\PsiFoo\in\Reals$ such that $\gamma_1(1)=RX(\PsiFoo)$.
Define $\gamma_2(s) = B_{-\PsiFoo} ( \gamma_1(s) ) - (0,R/2)$ where $B_{-\PsiFoo}$
is the boost of angle $-\PsiFoo$ about the origin.  Then
$\gamma_2(0)=(0,-R/2)$ and
\begin{equation}
\begin{aligned}
\gamma_2(1) &= B_{-\PsiFoo}(R X(\PsiFoo)) - (0,R/2) \\
&= (0,R) - (0,R/2)\\
&= (0,R/2).
\end{aligned}
\end{equation}
Figure \ref{fig:curve1} illustrates the construction of $\gamma_2$.

Let $\tau^\circ$ be defined by equation \eqref{eq:Theta} and suppose
$\tau^\circ>0$. Let  $\Psi = \ell\tau^\circ$
and let $T_1=T(-\Psi/2)$ and $T_2=T(\Psi/2)$. The lines through
$\gamma_2(0)$ and $\gamma_2(1)$ parallel to $T_1$ and $T_2$ respectively
meet at the point
\begin{equation}\label{eq:Q}
Q = (-(R/2)\,\coth(\Psi/2),0)
\end{equation}
and we define $\gamma_3 = \gamma_2 - Q$; see Figure \ref{fig:curve2}.

Let $\Pi$ be the projection from $I_+$ to $\calK_{\Psi}$.
We would like to show that $\gamma_3(s)$ lies in $I_+(0)$ for all $s$
so that we may form the composition $\Pi\circ\gamma_3$.  First, suppose 
$0< s < 1$.  Then $\gamma_3(s)-\gamma_3(0)$ is spacelike and rightward
pointing as it is an integral of such vectors.  Similarly,
$\gamma_3(s)-\gamma_3(1)$ is spacelike and leftward pointing.  So for all $s\in[0,1]$,
$\gamma_3(s)$ lies in the square with vertices
\begin{equation}
\begin{aligned}
P_1 &= \gamma_3(0) = \frac{R}{2}\left(\coth(\Psi/2),-1\right)\\
P_2 &= \gamma_3(1) = \frac{R}{2}\left(\coth(\Psi/2),1\right)\\
P_3 &= P_1 + (R/2,R/2) = P_2+(R/2,-R/2)= \frac{R}{2}\left(\coth(\Psi/2)+1,0\right)\\
P_4 &= P_1 + (-R/2,R/2) =  P_2+(-R/2,-R/2) = \frac{R}{2}\left(\coth(\Psi/2)-1,0\right).
\end{aligned}
\end{equation}
Each of these four points lies in the convex set $I_+(0)$, as does their convex hull.
Hence $\gamma_3(s)\in I_+(0)$ for all $s\in[0,1]$.

\begin{figure}
\ifjournal
\includegraphics{cc2_j}
\else\hskip -0.8cm
\includegraphics{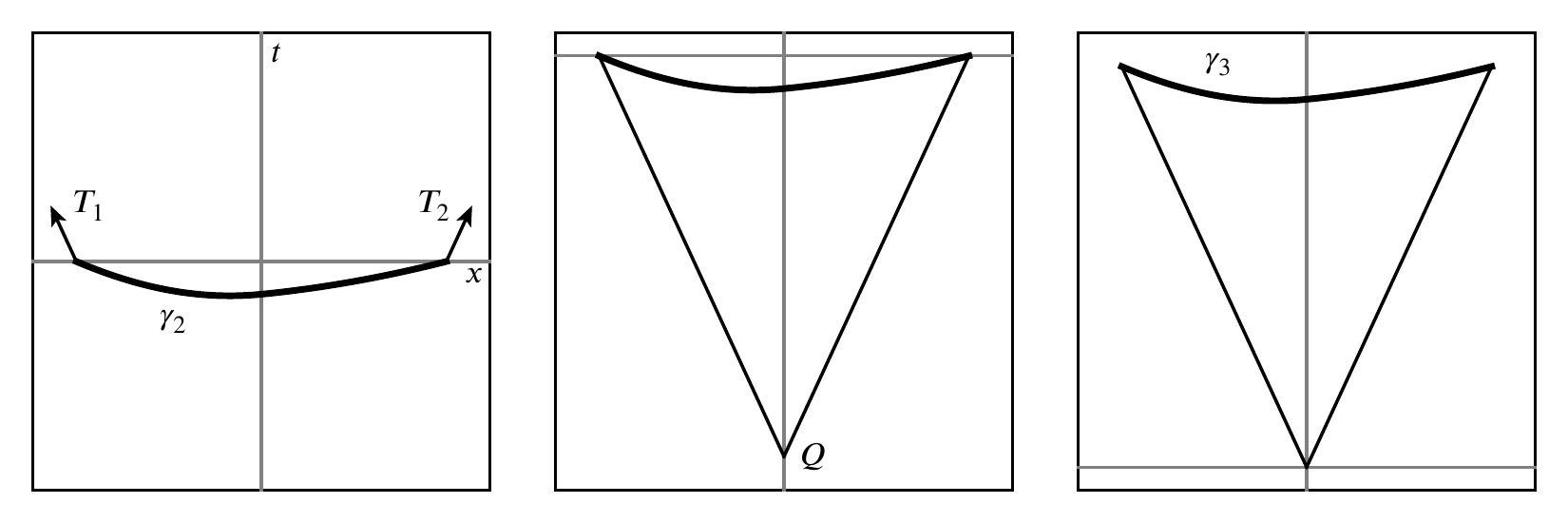}\fi
\caption{\label{fig:curve2}Transformation from $\gamma_2$ to $\gamma_3$: translate vertically
so that the lines from the curve endpoints to the origin meet at equal hyperbolic angles.}
\end{figure}

We now show that for all $k\in\Ints$,
$\gamma_3(s+ k) = B_{\Psi}^k \circ \gamma(s)$. Then, since every point
on $\gamma_3$ is the image under a boost of a point in $I_+(0)$, and 
since the set $I_+(0)$ is invariant under boosts about the origin, we will have 
shown that $\gamma_3(s)\in I_+(0)$ for all $s\in \Reals$.
In fact, an induction argument shows that it suffices to prove the fact for $k=1$; i.e.,
\begin{equation}
\gamma_3(s+1) = B_{\Psi} \circ \gamma(s)
\end{equation}
for all $s\in \Reals$.   To do this, let $\hat\gamma = B_{\Psi}\circ \gamma$.
A simple computation shows that
\begin{equation}
\gamma_3(0) = AT(-\Psi/2)\qquad\text{and}\qquad \gamma_3(1) = AT(\Psi/2)
\end{equation}
where $A=R/ (2\sinh(\Psi/2))$.
So
\begin{equation}
\hat\gamma(0) = B_{\Psi}(\gamma_3(0)) = B_{ \Psi}(A T(-\Psi/2)) = A T(\Psi/2) = \gamma_3(1).
\end{equation}
Moreover
\begin{equation}
\hat\gamma'(s) = B_{\Psi}(\gamma_3'(s)) = B_{\Psi} \ell X(\ell \theta(s) - \PsiFoo) = 
\ell X(\ell\theta(s)+\Psi - \PsiFoo).
\end{equation}
Now
\begin{equation}
\theta(s+1) = \ell\theta(s) + \ell\tau^\circ = \ell\theta(s) + \Psi
\end{equation}
by the $1$-periodicity of $\tau$ and the definition of $\Psi$.
Hence
\begin{equation}
\hat\gamma'(s) = \ell X(\ell\;\theta(s+1) - \PsiFoo) = \gamma'_3(s+1).
\end{equation}
Since $\hat\gamma(0)=\gamma(1)$ and $\hat\gamma'(s) = \gamma'(s+1)$ for all $s$,
we conclude that $B_{\Psi}\circ\gamma(s)=\hat\gamma(s) = \gamma(s+1)$ as claimed.

Since the image of $\gamma_3$ is contained in $I_+(0)$ we may form the projection
$\Pi\circ \gamma_3$ as desired.  Moreover, for all $s\in\Reals$ and $k\in \Ints$,
\begin{equation}
(\Pi\circ \gamma_3)(s+k) = (\Pi \circ B_{\Psi}^k)(\gamma_3(s)) = \Pi(\gamma_3(s)).
\end{equation}
Thus the curve $\gamma_3$ descends to a curve $\iota:S^1\ra \calK_{\Psi}$ such that
the following diagram commutes:
\begin{gather}
\xymatrix{\Reals \ar[d]_{\varepsilon} \ar[r]^{\gamma_3} & \Reals^{1,1}\ar[d]^{\Pi} \\
S^1 \ar[r]_{\iota} & \calK_{\Psi}\;, \\}
\end{gather}
where $\epsilon$ is the covering map mentioned in Section \ref{sec:prelims}.

It is easy to see that $\iota$ is an immersion; since $S^1$ is compact it is also an embedding.
Let $g$ and $K$ be the induced metric and second fundamental form of $\iota(S^1)$.
Since all of the transformations involved in this construction are local isometries,
to compute $g$ and $K$, it suffices to work
with the original curve $\gamma_1$.  Recall that $\gamma_1'(s)= \ell X(\ell\theta(s))$,
so $g((\epsilon\circ\iota)_*\partial_s,(\epsilon\circ\iota)_*\partial_s)=\ip<\gamma_1',\gamma_1'>_{R^{1,1}}=\ell^2$. Hence 
$\iota^* g = g_\ell$.
Moreover,
\begin{equation}
\gamma_1''(s) = \ell^2 T(\ell\theta(s))\theta'(s) = \ell^2 T(\ell\theta(s))\;\kappa(s)
\end{equation}
and the future pointing normal at $\gamma_1(s)$ is $T(\ell\theta(s))$.
So
\begin{equation}
\begin{aligned}
K((\epsilon\circ\iota)_*\partial_s,(\epsilon\circ\iota)_*\partial_s) &= 
-\ip<T(\ell\theta(s)),\gamma_1''>_{R^{1,1}}\\
& = 
-\ip<T(\ell\theta(s)),\ell^2 T(\ell\theta(s))\kappa(s)>_{R^{1,1}}\\
&=\tau(s)\ell^{2}
\end{aligned}
\end{equation}
and $K=\tau g_\ell$ as claimed. 
% The embedding has $W^{k+2,p}$ regularity since $\gamma_1''(s)=\ell^2 T(\ell\theta(s))\in W^{k,p}$.

This completes the proof in the case $\tau^\circ>0$, and the case $\tau^\circ<0$ is proved similarly (and
also follows from the $\tau^\circ>0$ case by a time reflection argument). 

When $\tau^\circ=0$, we can no longer compute the intersection point $Q$ from
equation \eqref{eq:Q} and the argument must change.  We simply let $\gamma_3=\gamma_2$
and replace the boost $B_{\Psi}$ with the operator $T_R$ that translates by $(0,R)$.  
Note that
\begin{equation}
T_R(\gamma_3(0)) = \gamma_3(0) + (0,R) = (0,-R/2) + (0,R) = (0,R/2) = \gamma_3(1).
\end{equation}
Moreover, since $\tau^\circ=0$, $\theta(s+1) = \theta(s)$ for all $s$ and
\ifjournal
\begin{equation}
\begin{aligned}
(T_R\circ\gamma_3)'(s) &= T_R(\gamma_3'(s)) = \gamma_3'(s) \\
&\qquad= \ell X(\ell\theta(s)-\PsiFoo) = \ell X(\ell(\theta(s+1)-\PsiFoo) = \gamma_3'(s+1).
\end{aligned}
\end{equation}
\else
\begin{equation}
(T_R\circ\gamma_3)'(s) = T_R(\gamma_3'(s)) = \gamma_3'(s) = \ell X(\ell\theta(s)-\PsiFoo) = 
\ell X(\ell(\theta(s+1)-\PsiFoo) = \gamma_3'(s+1).
\end{equation}
\fi
Hence $T_R\circ\gamma_3(s) = \gamma_3(s+1)$ for all $s\in\Reals$.  It then follows that
$T_R^k\circ\gamma_3(s) = \gamma_3(s+k)$ for all $k\in \Ints$ and $s\in\Reals$.

Let $\Pi$ be the projection into $\calC_{R}$.  Since the domain of $\Pi$ is all of $\Reals^{1,1}$, 
the image of $\gamma_3$ trivially lies in the domain of $\Pi$. We can then form the projection
$\Pi\circ \gamma_3$, and and the poof that this curve descends to a smooth embedding of $S^1$ into $\calC_R$
with the desired first and second fundamental forms now proceeds as before.
\end{proof}

For ease of presentation, we only considered smooth curvatures $\tau$ in Proposition \ref{prop:curve},
but it is easy to extend it to curvatures with lesser regularity.  For example, 
suppose $\tau\in L^\infty(S^1)$ instead.
Following the proof of Proposition \ref{prop:curve}, we find that the curve $\gamma_1$ is $C^{1,1}(S^1)$ (i.e., has
Lipschitz continuous first derivatives). The transformed curve $\gamma_3$ is also $C^{1,1}(S^1)$ 
and the remainder of the proof goes through without change, noting that the curvature of $\gamma_3$
is defined  weakly and pointwise almost everywhere by the usual formula \eqref{eq:coordcurvature}.   
Indeed, the curvature used to construct Figures \ref{fig:curve1} and \ref{fig:curve2} 
is among the $L^\infty$ mean curvatures of the type \eqref{eq:jump} 
considered in \cite{Maxwell:2011if}.
% $\tau(s)=1.4$ on half the circle and $\tau(s)=0.4$ on the other half.

\section{Simple Product Slices of Flat Kasner Spacetimes}\label{sec:non-cmc}

Consider a conformal data set on the torus
\begin{equation}\label{eq:udata}
(g_{\Bell}, \mu \sigma^\flat_\Bell, \tau, N)
\end{equation}
where $\Bell$ is a vector of lengths, $\mu$ is a constant,
$\sigma^\flat_\Bell$ is the transverse traceless tensor described
in equation \eqref{eq:sigflat}, and where the mean curvature $\tau$ and the densitized lapse $N$ are 
functions of the first coordinate $s^1$ of the torus only.  If $\tau$ is constant, then $N$ is irrelevant 
and the data set is exactly the type considered by Proposition \ref{prop:cmcflat}.
Thus it develops a flat Kasner, dual-to-flat Kasner, or static toroidal spacetime.
In this section we prove a partial generalization of this fact:
in many cases a non-CMC data set of the form \eqref{eq:udata} generates a flat Kasner 
or static toroidal spacetime as well.

The data set \eqref{eq:udata} has a $U^{n-1}$ symmetry, and it is
natural to seek a solution $(\phi,W)$ of the CTS-H equations
with the same symmetry. Writing $s$ and $\ell$ for $s^1$ and $\ell_1$,
suppose $\phi=\phi(s)$ and $W= w(s)\partial_{s}$
for some functions $\phi$ and $w$ on $S^1$.
A computation the shows that 
\begin{equation}\label{eq:sundires}
\begin{aligned}
-\Lap \phi &= \ell^2 \phi''\\
\ck W &= 2w' \sigma^\flat_\Bell \\
\mu \sigma^\flat_\Bell +\frac{1}{2N}\ck W &= \left(\mu+\frac{1}{N}w'\right) \sigma^\flat_\Bell \\
\div \left(\frac{1}{2N} \ck W\right) &= \left(\frac{1}{N} w'\right)' ds\\
d\tau &= \tau' ds
\end{aligned}
\end{equation}
where the primes denote derivatives with respect to $s$.  Using the fact that
\begin{equation}
\left|\sigma^\flat_\Bell\right|^2_{g_\Bell}=\kappa,
\end{equation}
the CTS-H equations read
\begin{subequations}\label{eq:CTS-H-reduced}
\begin{alignat}{2}
-2\kappa q \ell^2\, \phi'' - \kappa(\mu + (w'/N))^2\phi^{-q-1} + \kappa \tau^2\phi^{q-1} &= 0
\label{eq:CTS-H-red-ham} \\
(w'/N)' &= \phi^q \tau'. \label{eq:CTS-H-red-momentum}
\end{alignat}
\end{subequations}

Perhaps surprisingly, equations \eqref{eq:CTS-H-reduced} frequently admit solutions with a 
conformal factor $\phi=c$ for some constant $c$.  To see this, let us decompose
\begin{equation}
\tau(s) = \tau^* + \xi(s)
\end{equation}
where $\tau_0$ is a constant and 
\begin{equation}\label{eq:int_xi}
\int_{S^1} N(s) \xi(s)\; ds = 0;
\end{equation}
the constant $\tau^*$ is uniquely determined by
\begin{equation}
\tau^* = \frac{\int_{S^1} N \tau\; ds}{\int_{S^1} N \; ds}.
\end{equation}
Substituting $\phi=c$ into the
momentum constraint \eqref{eq:CTS-H-red-momentum} we find
\begin{equation}
\frac{w'}{N} = c^q (\tau^* + \xi + C)
\end{equation}
where the constant $C$ is determined by the condition
$\int_{S^1} w'\; ds = 0$.  We find that $C=-\tau^*$
and hence
\begin{equation}\label{eq:wpoverN}
\frac{w'}{N} = c^q \xi.
\end{equation}
Substituting $\phi=c$ and equation \eqref{eq:wpoverN} into the Hamiltonian constraint \eqref{eq:CTS-H-red-ham} we
obtain
\begin{equation}\label{eq:constsol}
-\kappa(\mu + c^q\xi)^2 c^{-q-1} + \kappa(\tau^* + \xi)^2c^{q-1} = 0.
\end{equation}
The constant $c$ is a solution of equation \eqref{eq:constsol} if and only if 
\begin{equation}
(c^{-q}\mu + \xi)^2 = (\tau^* + \xi)^2,
\end{equation}
or equivalently, if and only if
\begin{equation}\label{eq:d-mu-tau}
2(\tau^*-c^{-q}\mu)\xi = (c^{-2q}\mu^2-(\tau^*)^2).
\end{equation}
The right-hand size of equation \eqref{eq:d-mu-tau}
is constant, so if a solution $c$ exists, then either $\xi$
is constant or both sides vanish.  Note however that if $\xi$
is constant, the condition $\int N\xi ds=0$ ensures that $\xi\equiv0$
and $\tau\equiv \tau^*$.  Thus we find that there is a constant $c$ solving 
\eqref{eq:d-mu-tau} if and only if one of the following conditions holds
\begin{enumerate}
\item $\mu$ and $\tau^*$ are nonzero and have the same sign, in which case
\begin{equation}\label{eq:d-noncmc}
c = \left(\frac{\mu}{\tau^*}\right)^\frac{1}{q}.
\end{equation}
\item $\mu$ and $\tau^*$ are both zero, in which case any $c>0$ is a solution.
\item $\mu$ and $\tau^*$ are nonzero, have the opposite sign, and $\tau\equiv\tau^*$,
in which case
\begin{equation}
c = \left|\frac{\mu}{\tau^*}\right|^{\frac{1}{q}}.
\end{equation}
\end{enumerate}
Conformal data sets satisfying condition 3 also satisfy the conditions of the CMC construction
of Proposition \ref{prop:cmc} and generate a CMC slice of a dual-to-flat spacetime.
The first two conditions, however, are new, and we now examine the associated solutions
of the constraints.

The solution metric is
\begin{equation}
\begin{aligned}
\ol g &= c^{q-2} g_\Bell \\
&= g_{r\Bell}
\end{aligned}
\end{equation}
where 
\begin{equation}
r=c^{(q-2)/2} = \left( \frac{\mu}{\tau^*}\right)^\frac 1 n.
\end{equation}
To compute the second fundamental form of the solution we first note 
from equations \eqref{eq:sundires}, \eqref{eq:wpoverN}, as well as the fact 
that $c^q\tau^* = \mu$, that
\begin{equation}
\begin{aligned}
\mu \sigma^\flat_\Bell + \frac{1}{2N}\ck W &=
\mu \sigma^\flat_\Bell + \frac{w'}{N} \sigma^\flat_\Bell\\
&=(\mu  + c^q \xi) \sigma^\flat_\Bell\\
&=(c^q\tau^* + c^q \xi) \sigma^\flat_\Bell\\
&=c^q\tau \sigma^\flat_\Bell.
\end{aligned}
\end{equation}
Hence
\begin{equation}
\begin{aligned}
\ol K &= c^{-2}\left(\mu \sigma^\flat_\Bell + \frac{1}{2N}\ck W\right) + \frac{\tau}{n}c^{q-2}g_\Bell\\
&= c^{q-2}\tau(\sigma^\flat_\Bell + \frac{1}{n}g_\Bell)\\
&= r^2 \tau \ell^2 (ds)^2\\
&=\tau\; (r\ell)^2 (ds)^2.
\end{aligned}
\end{equation}
We summarize these computations as follows.
\begin{proposition}\label{prop:noncmc}
Consider a conformal data set $(g_\Bell, \mu \sigma^\flat_\Bell, \tau, N)$ on $T^n$
where $\mu$ is a constant and $\tau$ and $N$ are functions of the first factor of $T^n$ only.
Let
\begin{equation}\label{eq:tau_0}
\tau^* = \frac{\int_{S^1} N\tau\;ds}{\int_{S^1}N\; ds}
\end{equation}
where $s$ is the unit coordinate on $S^1$ and let $\xi = \tau - \tau^*$.  

If $\mu$ and $\tau^*$ are nonzero and have the same sign, then
there is a solution $(\phi, W)$  of the CTS-H equations such that
\begin{equation}
\phi\equiv c = \left(\frac{\mu}{\tau^*}\right)^\frac{1}{q}.
\end{equation}
The vector field $W$ is parallel to $\partial_{s^1}$ and
\begin{equation}\label{eq:ckW}
\frac{1}{2N} \ck W = c^q \xi \sigma^\flat_\Bell.
\end{equation}
The solution of the constraint equations generated by this solution of the CTS-H 
equations is
\begin{equation}\label{eq:nonCMCflat}
\begin{aligned}
\ol g &= g_{r\Bell} \\
\ol K &= \tau\; (r\ell_1)^2 (ds^1)^2
\end{aligned}
\end{equation}
where 
\begin{equation}
r=c^{(q-2)/2}=\left(\frac{\mu}{\tau^*}\right)^{\frac{1}{n}}.
\end{equation}

If $\mu$ and $\tau^*$ are both equal to zero, then for every
$c>0$ there is a solution $(\phi, W)$ of the CTS-H equations such that
$\phi\equiv c$ everywhere and such that equation
\eqref{eq:ckW} holds.  The associated solution of the constraint equations is
given by equations \eqref{eq:nonCMCflat} with $r=c^{(q-2)/2}$.
\end{proposition}

When $\tau^*\neq 0$, Proposition \ref{prop:noncmc} provides the existence of at least
one solution of the CTS-H equations, but it does not exclude the existence of others.  One
would expect that, at least under a near-CMC condition, there is only one.  However, the standard
near-CMC results do not apply here because of the conformal Killing fields on the flat torus,
and it remains to be seen if uniqueness holds when $\tau^*\neq 0$.
The one-parameter family of solutions when $\tau^*=0$ is, in the non-CMC case, an unusual 
feature of the conformal method.  The condition $\tau^*=0$ is equivalent to
\begin{equation}
\int_{S^1} N \tau\; ds = 0
\end{equation}
and we see that it can only hold if $\tau\equiv 0$ or if $\tau$ changes signs.
Although the literature contains a variety of notions of what constitutes a
near-CMC condition for the conformal method, they all
exclude the possibility that the mean curvature changes signs. Thus the one-parameter families
constructed here are a genuinely far-from-CMC phenomenon.  It
is important to keep in mind
that the condition $\tau^*=0$ depends not only on the mean curvature, but also
on the choice of densitized lapse.  If $\tau$ changes signs, one can always find a 
densitized lapse
such that $\tau^*=0$, but for a generic densitized lapse we will have $\tau^*\neq 0$.

Proposition \ref{prop:noncmc} only considers the cases where $\tau^*$ and $\mu$ have
the same sign, and is silent otherwise.  In Proposition \ref{prop:cmcflat},
the CMC version of Proposition \ref{prop:noncmc}, the condition
that $\tau^*$ and $\mu$  have the same sign is precisely the one 
that ensures that the generated solutions of the constraints are slices of
flat Kasner or static toroidal spacetimes.  We now wish to extend
this result to Proposition \ref{prop:noncmc}.

Consider a flat Kasner spacetime
\begin{equation}
\mathcal M = \mathcal K_\Psi \times T^{n-1}_{\hat\Bell}.
\end{equation}
An embedding $\iota: T^{n} \ra \mathcal M$ is called
a \textbf{simple product embedding} if there is a curve 
$\gamma: S^1 \ra \mathcal K_\Psi$ such that
\begin{equation}
\iota(s^1,\ldots,s^n) = \gamma(s^1)\times (s^2,\ldots,s^n);
\end{equation}
simple product embeddings into static toroidal spacetimes are defined similarly.
The solutions of the constraints 
constructed in Proposition \ref{prop:noncmc} are the simple product slices of flat Kasner or 
static toroidal spacetimes.

\begin{proposition}\label{prop:class}
Let 
\begin{equation}\label{eq:theta2}
\tau^\circ = \int_{S^1} \tau\; ds
\end{equation}
where here and in the following we continue with the notation of 
Proposition \ref{prop:noncmc}. 

If $\tau^\circ\neq 0$ then the metric
and second fundamental form from equations \eqref{eq:nonCMCflat}
are induced by a simple product embedding into the flat Kasner spacetime
\begin{equation}
\mathcal K_{\Psi} \times T^{n-1}_{\hat \Bell}
\end{equation}
where 
\begin{equation}
\Psi = (r\ell_1)\tau^\circ
\end{equation}
and
\begin{equation}
\hat\Bell = (r\ell^2,\ldots,r\ell^n).
\end{equation}
If $\tau^\circ = 0$ then 
the metric
and second fundamental form from equations \eqref{eq:nonCMCflat}
are induced by a simple product embedding into a static toroidal spacetime
\begin{equation}
\mathcal C_{L}  \times T^{n-1}_{\hat \Bell}
\end{equation}
for some $L>0$.
\end{proposition}
\begin{proof}
Suppose that $\tau^\circ\neq 0$.  Let $\Psi = (r\ell_1)\tau^\circ$ and
let $\gamma:S^1\ra \mathcal K_\Psi$ be the curve given by Proposition
\ref{prop:curve} such that its inherited metric is $(r\ell_1)^2 (ds)^2$
and its second fundamental form is $\tau\; (r\ell_1)^2 (ds)^2$.
We define
\begin{equation}
\iota: T^n \ra \mathcal K_\Psi \times T^{n-1}_{\hat\Bell}
\end{equation}
by
\begin{equation}
\iota(s^1,\ldots, s^n) = \gamma(s^1) \times (s^2,\ldots, s^n).
\end{equation}
The induced metric on $T^n$ is
\begin{equation}\label{eq:myg}
\begin{aligned}
\ol g &= (r\ell_1)^2 (ds^1)^2 + \left[ 
(r\ell_2)^2 (ds^2)^2 + \cdots + (r\ell_n)^2 (ds^n)^2\right]\\
& = g_{r\Bell}
\end{aligned}
\end{equation}
and the second fundamental form is
\begin{equation}
\begin{aligned}\label{eq:myK}
\ol K &= \tau\; (r\ell_1)^2 (ds^1)^2 + 0 \\
&= \tau\; (r\ell_1)^2 (ds^1)^2. 
\end{aligned}
\end{equation}
Comparing equations \eqref{eq:myg} and \eqref{eq:myK} with
equations \eqref{eq:nonCMCflat} completes the proof in the case
$\tau^\circ\neq 0$.

The case where $\tau^\circ=0$ is proved similarly, replacing the
curve into $\mathcal K_\Psi$ with a curve into $\mathcal C_{L}$ for some $L>0$,
as provided by Proposition \ref{prop:curve}.
\end{proof}

Proposition \ref{prop:class} shows that initial data of the type considered in 
Proposition \ref{prop:noncmc} give rise to simple product embeddings into
flat Kasner or static toroidal spacetimes.  It is not hard to see that every
simple product embedding arises from some choice of this data.  Indeed,
the metric and second fundamental form for such an embedding will have the
form
\begin{equation}\label{eq:flatsol}
\begin{aligned}
\ol g &= g_{\Bell}\\
\ol K &= \tau\; (\ell_1)^2 (ds^1)^2
\end{aligned}
\end{equation}
for some vector of lengths $\Bell$ and some mean curvature $\tau$. 
Let $N$ be an arbitrary densitized lapse depending only on $s^1$ and
define $\tau^*$ by equation \eqref{eq:tau_0}.
Proposition \ref{prop:noncmc} then implies that
a conformal data set $(g_\Bell, \tau^* \sigma^\flat_\Bell, \tau, N)$
leads to the solution \eqref{eq:flatsol} of the constraint equations.
Thus Proposition \ref{prop:noncmc} gives a complete description of
the construction, via the conformal method, of simple product slices of 
flat Kasner and static toroidal spacetimes.

From Propositions \ref{prop:noncmc} and \ref{prop:class} we see the
importance of the quantities
\begin{equation}
\tau^* = \frac{\int_{S^1} N \tau \; ds}{\int_{S^1} N\; ds}
\end{equation}
and
\begin{equation}
\tau^\circ = \int_{S^1} \tau \; ds = \frac{\int_{S^1} \tau \; ds} {\int_{S^1} 1\; ds}.
\end{equation}
The value of $\tau^*$, in combination with the parameter $\mu$, determines 
if and how many flat spacetimes will be constructed.  If $\tau^*$ and $\mu$
have the same sign and are non-zero, there will be exactly one flat spacetime, 
and if they both vanish there will be a one-parameter family.
The parameter $\tau^\circ$ determines what type of flat spacetime will be constructed.
If $\tau^\circ>0$ it will be an expanding flat Kasner solution, if $\tau^\circ=0$ it will be static
toroidal, and if $\tau^\circ<0$ it will be a contracting flat Kasner solution.

To visualize this situation it is helpful to make a diagram analogous to Figure
\ref{fig:cmcflat}.  Consider a conformal data set
\begin{equation}
(g_{\Bell}, \mu \sigma_{\Bell}^\flat, \tau_a=a+\xi, N)
\end{equation}
where $N$ and $\xi$ are fixed functions and $\mu$ and $a$ are constants that
we will adjust.  The parameter $\mu$ controls the size of the TT tensor, and $a$
controls whether the data is near-CMC or not, with $|a|$ large relative to $\osc \xi$ 
being a near-CMC condition.  
In the CMC case, we could take $\xi\equiv 0$, in which case we would arrive at the diagram of
Figure \ref{fig:cmcflat} with $a$ in place of $\tau_0$.

\begin{figure}
\hfil\includegraphics{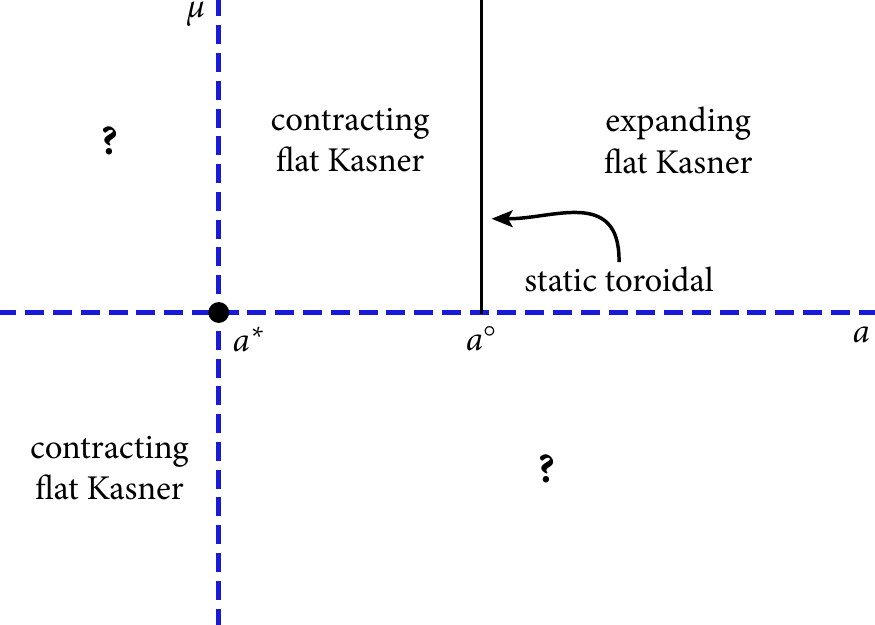}
\caption{\label{fig:noncmcflat}Conformal method construction for a data set $(g_{\Bell}, \mu \sigma_{\Bell}^\flat, a+\xi, N)$
in the case $a^* < a^\circ$.
Dashed blue lines are the boundaries of regions
where the number and type of solution is unknown, and
the dot at $\mu=0$, $a=a^*$ corresponds to a 
one-parameter family of contracting flat Kasner solutions. Compare with Figure \ref{fig:cmcflat}.}
\end{figure}

If $a$ is positive and large enough so that $\tau_a = a+\xi$ is positive, then
both $\tau^*$ and $\tau^\circ$ will be positive.  So if $\mu$ is also positive,
then we will construct a single expanding flat Kasner spacetime from the
conformal data.  Similarly, if $a$ is negative and $|a|$ is large enough so that
$a+\xi$ is negative, then $\tau^*$ and $\tau^\circ$ will be negative, and
if $\mu$ is also negative we will construct a single contracting flat Kasner spacetime.
There are two special intermediate values of $a$; let
\begin{equation}
a^* = -\frac{\int_{S^1} N \xi \; ds}{\int_{S^1} N\; ds}
\end{equation}
and let
\begin{equation}
a^\circ = -\frac{\int_{S^1} \xi \; ds}{\int_{S^1} 1\; ds}.
\end{equation}
When $a=a^*$ then $\tau^*=0$ and we pick up a one-parameter family of solutions (so long as $\mu=0$ as well).
When $a=a^\circ$, then $\tau^\circ=0$ and for this choice of $a$ we construct static toroidal solutions rather
than flat Kasner solutions; if $a>a^\circ$ we construct expanding flat Kasner solutions and if $a<a^\circ$ we construct contracting flat Kasner solutions.  
The situation is summarized in Figure \ref{fig:noncmcflat} in the case where $a^*< a^\circ$.

In the quadrants in Figure \ref{fig:noncmcflat} with question marks, $\mu$ and $\tau^*$
have opposite signs and therefore Proposition \ref{prop:noncmc} does not apply.
We do not know how many solutions may exist in this case, but we do know that
if a solution exists, it will not be a simple product slice of a flat Kasner or static
toroidal spacetime.  Moreover, one readily verifies that the volumes of the slices constructed
by Proposition \ref{prop:noncmc} are proportional to $\mu/\tau^*$.
So the volume of the solution metric 
is constant on every line of positive slope through the dot in Figure \ref{fig:noncmcflat},
and is indeed proportional to the slope of the line.  The horizontal dashed line corresponds to zero
volume and the vertical dashed line corresponds to infinite volume.

\subsection{Relationship with the construction from \cite{Maxwell:2011if}}\label{sec:jump}

The study in \cite{Maxwell:2011if} considers conformal data sets
\begin{equation}
(g_\Bell, \sigma, \tau, N)
\end{equation}
on  $T^n$ where $N$ is a function of $s^1$ only,
$\tau$ is a member of a particular family of mean curvatures depending on $s^1$ only,
and $\sigma$
is an arbitrary TT-tensor with vanishing Lie derivative in the $s^2$ through $s^k$ directions.
It was shown in \cite{Maxwell:2011if} that the role of such a $\sigma$ in the CTS-H equations
can be reduced to a function $\eta(s^1)$ and a constant $\mu$, and the 
CTS-H equations for solutions with the same $U^{n-1}$ symmetry
as this data take the form 
\begin{subequations}\label{eq:CTS-H-reduced-11}
\begin{alignat}{2}
-2\kappa q \ell^2\, \phi'' - 2\eta^2  -\kappa(\mu + (w'/N))^2\phi^{-q-1} + \kappa \tau^2\phi^{q-1} &= 0\\
(w'/N)' &= \phi^q \tau'.
\end{alignat}
\end{subequations}
These are equivalent to equations (20) from \cite{Maxwell:2011if} (with $w/2$ in \cite{Maxwell:2011if} 
corresponding with $w$ in this paper).  Comparing equations \eqref{eq:CTS-H-reduced-11} with equations
\eqref{eq:CTS-H-reduced} we see that the TT-tensors considered in Propositions \ref{prop:noncmc}
and \ref{prop:class} correspond to $\eta\equiv 0$ in \cite{Maxwell:2011if}.  The mean curvatures
from \cite{Maxwell:2011if} have the specific form $\tau_a(s) =a + \xi_{\rm jump}$ where $a$ is a constant and
\begin{equation}
\xi_{\rm jump}(s) = \begin{cases} 1  & 0 < s < \frac 1 2\\
-1 & \frac 1 2 < s < 1.\end{cases}
\end{equation}
Thus the overlap between the two papers corresponds either
to $\eta\equiv 0$ in \cite{Maxwell:2011if},
or equivalently to conformal data sets of the type
\begin{equation}\label{eq:overlapdata}
(g_\Bell, \mu \sigma^\flat_\Bell, \tau_a, N)
\end{equation}
where $N$ depends only on $s^1$ and $\tau_a = a + \xi_{\rm jump}$. 

Most of the interesting results from \cite{Maxwell:2011if} were
obtained under the hypothesis $\mu=0$ and $\eta$ is arbitrary.
This is in contrast to the situation here where $\mu$
is arbitrary and $\eta\equiv 0$.  Nevertheless, it is instructive
to compare the results between the two papers on their
common region of parameter space. To do this, we first
observe that
\begin{equation}\label{eq:overlapastar}
a^* = - \frac{\int_{S^1}N\xi_{\rm jump}\; ds}{\int_{S^1}N\; ds}
\end{equation}
also appears in \cite{Maxwell:2011if}, where it was
denoted $\gamma_N$.  With this change of notation, we translate 
Theorems 3.1 through 3.3 of \cite{Maxwell:2011if}, under the assumption that $\eta\equiv0$, as
follows.
\begin{proposition}[\cite{Maxwell:2011if}]\label{prop:overlap1}
Consider a conformal data set of the form \eqref{eq:overlapdata} and let $a^*$
be defined by \eqref{eq:overlapastar}.
\begin{itemize}
\item If $|a|>1$ (i.e., if $\tau_a$ has constant sign) then there is at least one solution of system \eqref{eq:CTS-H-reduced-11}.
\item If $|a^*|<|a|<1$ and if $\mu$ is sufficiently small, with smallness depending on $a$, there is at least
one solution of system \eqref{eq:CTS-H-reduced-11}.
\item If $a=a^*$ and $\mu=0$ then there is a one-parameter family of solutions.
\end{itemize}
\end{proposition}

To apply the results of the current paper to this same conformal data set, 
we first recall that Proposition \ref{prop:curve} applies equally well to $L^\infty$
curvatures. It is then easy to verify that, without changing the proofs, we also obtain Propositions
\ref{prop:noncmc} and \ref{prop:class} for $L^\infty$
curvatures. Hence we obtain the following.
\begin{proposition}\label{prop:overlap2}
Consider a conformal data set of the form \eqref{eq:overlapdata} and let $a^*$
be defined by \eqref{eq:overlapastar}.
\begin{itemize}
\item If $a-a^*$ and $\mu$ are non-zero and have the same sign, there is at least one
solution (and it generates either a flat Kasner or a static toroidal spacetime as 
described by Proposition \ref{prop:class}).
\item If $a=a^*$ and $\mu=0$, then there is a one parameter family of
solutions (and they generates either a family of flat Kasner or static toroidal spacetimes as 
described by Proposition \ref{prop:class}).
\end{itemize}
\end{proposition}
The existence results of Propositions \ref{prop:overlap1} and \ref{prop:overlap2} 
have nontrivial intersection, but each nontrivially extends the other.  
Proposition \ref{prop:overlap2} does not consider the case where $a-a^*$ and $\mu$
have opposite signs, whereas  Proposition \ref{prop:overlap2} does.  On the other hand,
when $a-a^*$ and $\mu$ have the same nonzero sign, Proposition \ref{prop:overlap2}
provides a solution, whereas Proposition \ref{prop:overlap1} requires the
additional hypothesis that $\mu$ is small if $|a|<1$.
The region where no existence result is obtained by either paper is restricted to
the following regions:
\begin{itemize}
\item $\mu=0$ or $a=a^*$, but not both,
\item $a-a^*$ and $\mu$ have opposite signs, and either
$|a|< |a^*|$, or  $|a^*|<|a|<1$, and $|\mu|$ is large.
\end{itemize}
It would be reasonable to conjecture that in the first case there is no solution, whereas
in the second case there is a solution, in which case we would obtain a complete analogue
of Proposition \ref{prop:cmcflat} in this setting.

Perhaps most interestingly, we find that in the regions of parameter space 
where both theorems predict existence, the resulting spacetimes are smooth, 
even though the mean curvatures are only $L^\infty$.  Moreover, we
find that the one-parameter family discovered in Proposition \ref{prop:overlap1} is not an artifact
of low regularity, but is simply a special case of the more general family found in
Proposition \ref{prop:noncmc}.

\section{Conformal Transformations of the Mean Curvature}\label{sec:mc-alt}

Although we have given a full description of the conformal construction of simple product 
slices of flat Kasner spacetimes, the quantities $\tau^*$ and $\tau^\circ$ that play such
a key role in the description are difficult to detect from arbitrary conformal data sets.  

Consider the definition
\begin{equation}
\tau^*  = \frac{\int_{S^1} N \tau \; ds}{\int_{S^1} N\; ds}.
\end{equation}
Since $s$ is the unit coordinate on $S^1$, this definition
depends on a specific choice of coordinates.  This situation
can be rectified
somewhat by recalling that the full
conformal data set is $(g_\Bell, \mu\sigma^\flat_\Bell, \tau, N)$  
and therefore 
\begin{equation}
\int_{T^n} N\tau \; dV_\Bell = (\ell_1\cdots\ell_n) \int_{S^1} N\tau\; ds
\end{equation}
and similarly 
\begin{equation}
\int_{T^n} N \; dV_\Bell = (\ell_1\cdots\ell_n)  \int_{S^1} N\; ds.
\end{equation}
Hence
\begin{equation}\label{eq:taustartinv}
\tau^* = \frac{\int_{T^n} N \tau\; dV_\Bell}{\int_{T^n} N\; dV_\Bell},
\end{equation}
which is an expression that does not depend on coordinates.  However,
equation \eqref{eq:taustartinv} still suffers from the fact that it
lacks conformal covariance.  Let $\psi$ be a positive conformal factor
and let $\tilde g = \psi^{q-2}g_\Bell$, 
$\tilde \sigma = \psi^{-2}\mu\sigma^\flat_\Bell$ and $\tilde N = \psi^q N$.
The conformal data set $(\tilde g, \tilde \sigma, \tau, \tilde N)$
generates the same solutions of the constraint equations as 
$(g_\Bell, \mu\sigma^\flat_\Bell, \tau, N)$, and we would like to know if $\tau^*=0$
or not to determine if we will construct a one-parameter family.  Using
the fact that $\widetilde {dV} = \psi^q dV_\Bell$ equation \eqref{eq:taustartinv}
becomes
\begin{equation}\label{eq:taustartinv2}
\tau^* = \frac{\int_{T^n} \psi^{-2q} \tilde N \tau\; \widetilde{dV}}
{\int_{T^n} \psi^{-2q} \tilde N\; \widetilde{dV}}.
\end{equation}
A similar exercise shows that
\begin{equation}
\tau^\circ = \frac{\int_{T^n} \psi^{-q} \tau\; \widetilde{dV}}
{\int_{T^n} \psi^{-q} \; \widetilde{dV}}.
\end{equation}
It seems impossible to write $\tau^*$ and $\tau^\circ$ in terms of the data with tildes
without explicit reference to the conformal factor $\psi$.  Since the conformal factor solving the CTS-H equations
for the data with tildes is $\psi^{-1} c$ for some constant $c$, it appears
that computing $\tau^*$ and $\tau^\circ$  is tantamount to solving the CTS-H equations
in the first place.  We cannot determine if we will generate a one-parameter of solutions
for the data until we have effectively already solved the CTS-H equations.

On the other hand, for a conformal data set
$(g_\Bell, \mu\sigma^\flat_\Bell, \tau, N)$ leading to a flat spacetime, the
solution metric is $\ol g = c^{q-2} g_\Bell$ for some constant $c>0$
and the associated lapse is $\ol N = c^q N$. We
can therefore also write
\begin{equation}
\tau^* = \frac{\int_{T^n} c^{-2q} \ol N \tau\; \ol{dV}}
{\int_{T^n} c^{-2q} \ol N\; \ol{dV}} = 
\frac{\int_{T^n}  \ol N \tau\; \ol{dV}}
{\int_{T^n}  \ol N\; \ol{dV}}
\end{equation}
and similarly
\begin{equation}\label{eq:taustartinv3}
\tau^\circ = \frac{\int_{T^n} \tau\; \ol{dV}}
{\int_{T^n} \; \ol{dV}}
\end{equation}
where $\ol{dV}$ is the volume element of the solution metric.
Hence $\tau^*$ and $\tau^\circ$ are naturally computed in terms of the solution
variables $\ol g$ and $\ol N$.

We are therefore lead to consider alternative treatments of the mean curvature that allow
one to specify either $\tau^*$ or $\tau^\circ$ explicitly.  We focus our attention here on $\tau^*$
since the property it controls (whether we construct just one solution or a one-parameter
family of solutions) seems more essential than the property that $\tau^\circ$ controls (whether the spacetime
is an expanding or contracting Kasner).  Nevertheless, the approaches suggested below have obvious parallels
that allow one to specify $\tau^\circ$ as part of the conformal data set (at the expense of leaving $\tau^*$
unknown).

\subsection{Lapse-scaled mean curvature}
Starting with a traditional conformal data set $(g,\sigma, \tau, N)$ on a compact manifold $M^n$,
we decompose $\tau = \tau^* + \xi$ where
$\tau^*$ is a constant and $\int_M \xi N \; dV_g = 0$ to arrive at
a data set $(g,\sigma,\tau^*,\xi, N)$.  Given a conformal factor $\psi$ we now define 
a conformally related data set
\begin{equation}
(\tilde g, \tilde \sigma, \tau^*, \tilde \xi, \tilde N) = (\psi^{q-2} g, \psi^{-2}\sigma, \tau^*, \psi^{-2q}\xi, \psi^q N)
\end{equation}
and let the conformally related mean curvature be $\tilde \tau = \tau^* + \tilde\xi = \tau^* + \psi^{-2q}\xi$.
With this transformation law, using the fact that $\int_M N\xi\; dV = 0$, we find
\begin{equation}
\begin{aligned}
\frac{\int_M \tilde N \tilde \tau\; \widetilde{dV}} {\int_M \tilde N\;  \widetilde{dV}}
&= \frac{\int_{M} \psi^q N (\tau^* + \psi^{-2q} \xi) \psi^q\; dV}{\int_{M} \psi^q N \psi^q\; dV}\\
&=\frac{\int_{M} (\psi^{2q} N \tau^* +  N\xi)\; dV}{\int_{M} \psi^{2q} N\; dV}\\
&= \tau^*
\end{aligned}
\end{equation}
regardless of the  conformal factor $\psi$. So we are introducing a non-trivial conformal 
transformation law for $\tau$ that preserves the quantity $\tau^*$, and we note that it reduces
to the traditional law of unchanging mean curvature in the CMC case.  

Starting from a data set $(g,\sigma,\tau^*,\xi, N)$ we seek a solution of the constraints
\begin{equation}
\begin{aligned}
\ol g &= \phi^{q-2} g\\
\ol K &= \phi^{-2} \left( \sigma + \frac{1}{2N}\ck W\right) + \frac{\tau^* + \phi^{-2q}\xi}{n} \ol g.
\end{aligned}
\end{equation}
for some conformal factor $\phi$ and vector field $W$.  Substituting these equations into the constraint
equations we obtain
\ifjournal
\begin{equation}\label{eq:CTSHSTAR}
\begin{aligned}
-2\kappa q\Lap \phi + R\phi 
&= \left| \sigma + \frac{1}{2N} \ck W\right|^2\phi^{-q-1} - \kappa(\tau^*+\phi^{-2q}\xi)^2\phi^{q-1}\\
-\div\left(\frac{1}{2N}\ck W\right) &= \kappa \phi^q \;d (\phi^{-2q} \xi),
\end{aligned}
\end{equation}
\else
\begin{equation}\label{eq:CTSHSTAR}
\begin{aligned}
-2\kappa q\Lap \phi + R\phi - \left| \sigma + \frac{1}{2N} \ck W\right|^2\phi^{-q-1} + \kappa(\tau^*+\phi^{-2q}\xi)^2\phi^{q-1}
&= 0\\
-\div\left(\frac{1}{2N}\ck W\right) &= \kappa \phi^q \;d (\phi^{-2q} \xi),
\end{aligned}
\end{equation}
\fi
which we will call the \textbf{CTS-H}$^*$ equations.

The lapse-scaled mean curvature is a significant modification of the CTS-H equations, 
and we delay a general analysis of its properties for future work.  For now, we restrict our attention to
conformal data sets of the type considered in the non-CMC construction
of Proposition \ref{prop:noncmc}.  Consider a conformal data set
\begin{equation}
(g_\Bell, \mu \sigma^\flat_\Bell, \tau^*, \xi, N)
\end{equation}
on the torus $T^n$ where $\xi$ and $N$ are functions of $s^1$ alone.  
Writing $s=s^1$  we 
seek a solution of system \eqref{eq:CTSHSTAR} of the form $\phi=\phi(s)$
and $W = w(s) \partial_s$. The CTS-H$^*$ equations can then be rewritten
\begin{subequations}\label{eq:CTS-H-STAR-reduced}
\begin{alignat}{2}
-2\kappa q \ell^2\, \phi'' - \kappa(\mu + (w'/N))^2\phi^{-q-1} + \kappa (\tau^*+\phi^{-2q}\xi)^2 \phi^{q-1} &= 0
\label{eq:CTS-H-STAR-red-ham} \\
(w'/N)' &= \phi^q (\phi^{-2q}\xi)'. \label{eq:CTS-H-STAR-red-momentum}
\end{alignat}
\end{subequations}
where $\ell=\Bell_1$ and derivatives are taken with respect to $s$. 
An analysis identical to the one following equations \eqref{eq:CTS-H-reduced}
shows that if $\mu$ and $\tau^*$ are nonzero and have the same sign, then there is
a solution $(\phi,w)$ of system \eqref{eq:CTS-H-STAR-reduced} satisfying
\begin{equation}
\begin{aligned}
\phi \equiv c &= \left(\frac{\mu}{\tau^*}\right)^\frac{1}{q}\\
\frac{w'}{N} &= c^{-q}\xi.
\end{aligned}
\end{equation}
If $\mu$ and $\tau^*$ both vanish then $\phi\equiv c$ and $w'/N = c^{-q}\xi$ is a solution
for every $c>0$.  A parallel computation to the one that led to Proposition \ref{prop:noncmc}
now leads to the following.

\begin{proposition}\label{prop:noncmc-star}
Consider a CTS-H$^*$ conformal data set $(g_\Bell, \mu \sigma^\flat_\Bell, \tau^*, \xi, N)$ on $T^n$
where $\mu$ is a constant, $\xi$ and $N$ are functions of the first factor of $T^n$ only,
and $\int_{T^n} \xi N\; dV_\Bell = 0$.

If $\mu$ and $\tau^*$ are nonzero and have the same sign, then
there is a solution $(\phi, W)$  of the CTS-H$^*$ equations 
\eqref{eq:CTSHSTAR} such that
\begin{equation}
\phi\equiv c = \left(\frac{\mu}{\tau^*}\right)^\frac{1}{q}.
\end{equation}
The vector field $W$ is parallel to $\partial_{s^1}$ and
\begin{equation}\label{eq:ckW2}
\frac{1}{2N} \ck W = c^{-q} \xi \sigma^\flat_\Bell.
\end{equation}
The solution of the constraint equations generated by this solution of the CTS-H$^*$ 
equations is
\begin{equation}\label{eq:nonCMCflat-star}
\begin{aligned}
\ol g &= g_{r\Bell} \\
\ol K &= \ol\tau\; (r\ell_1)^2 (ds^1)^2
\end{aligned}
\end{equation}
where 
\begin{equation}
\ol \tau = \tau^* + c^{-2q} \xi
\end{equation}
and where
\begin{equation}
r=c^{(q-2)/2}=\left(\frac{\mu}{\tau^*}\right)^{\frac{1}{n}}.
\end{equation}

If $\mu$ and $\tau^*$ are both equal to zero, then for every
$c>0$ there is a solution $(\phi, W)$ of the CTS-H$^*$ equations such that
$\phi\equiv c$ everywhere and such that equation
\eqref{eq:ckW2} holds.  The associated solution of the constraint equations is
given by equations \eqref{eq:nonCMCflat} with $r=c^{(q-2)/2}$.
\end{proposition}

We also have an analogue of Proposition \ref{prop:class} that shows
that the solutions of the constraint equations generated by
Proposition \ref{prop:noncmc-star} are simple product slices
of flat Kasner and static toroidal spacetimes.

\begin{proposition}\label{prop:class-star}
Suppose $(\ol g, \ol K)$ is a solution of the constraints
generated by Proposition \ref{prop:noncmc-star} from CTS-H$^*$ data
\begin{equation}
(g_\Bell, \mu \sigma^\flat_\Bell, \tau^*, \xi, N)
\end{equation}
and let
\begin{equation}\label{eq:tau_circ_star}
\tau^\circ = \frac{\int_{T^n}\ol \tau\; dV_{\ol g}}{\int_{T^n}1\; dV_{\ol g}}
\end{equation}
where $\ol \tau$ is defined in Proposition \ref{prop:noncmc-star}.

If $\tau^\circ\neq 0$, then $\ol g$ and $\ol K$
are induced by a simple product embedding into the flat Kasner spacetime
\begin{equation}
\mathcal K_{\Psi} \times T^{n-1}_{\hat \Bell}
\end{equation}
where 
\begin{equation}
\Psi = (r\ell_1)\tau^\circ,
\end{equation}
\begin{equation}
\hat\Bell = (r\ell^2,\ldots,r\ell^n),
\end{equation}
and where $r=c^{(q-2)/2}$. If $\tau^\circ = 0$ then 
the metric
and second fundamental form from equations \eqref{eq:nonCMCflat}
are induced by a simple product embedding into a static toroidal spacetime
\begin{equation}
\mathcal C_{L}  \times T^{n-1}_{\hat \Bell}
\end{equation}
for some $L>0$.
\end{proposition}
\begin{proof}
The proof is identical to that of Proposition \ref{prop:class}
so long as we can show that $\tau^\circ$ defined by
equation \eqref{eq:tau_circ_star} is equal
to
\begin{equation}
\int_{S^1} \ol \tau\; ds
\end{equation}
where $s$ is the unit coordinate on the circle.  But this was
shown previously in the discussion leading to equation 
\eqref{eq:taustartinv3}.
\end{proof}

The moral of Propositions \ref{prop:noncmc-star} and \ref{prop:class-star} is
that, as far as the parameterization of flat Kasner spacetimes is concerned, the 
conformal method using a lapse-scaled
mean curvature
has essentially the same properties as the standard conformal method, except that we can predict
in advance when we will generate a one-parameter family of solutions.  These families occur when
$\mu=0$ and $\tau^*=0$, regardless of which background metric we use to represent the
conformal data.  Once a flat spacetime has been constructed, we can determine 
if it is an expanding or contracting Kasner solution via Proposition \ref{prop:class-star},
but we do not get to choose this property of the solution in advance.

\subsection{Drift-parameterized mean curvature}

In developing the lapse-scaled mean curvature we worked with a piece $\xi$ of the mean curvature that conformally 
transforms as $\tilde \xi = \phi^{-2q}\xi$ 
so that $\int_M \tilde N \tilde \xi\; \widetilde{dV}=0$.  As an alternative parameterization, we specify
a vector field $Z$ and write
\begin{equation}
\xi = \frac{1}{N} \div_g Z
\end{equation}
and
\begin{equation}
\tilde \xi = \frac{1}{\tilde N} \div_{\tilde g} Z = \phi^{-2q} \frac{1}{N} \div_{g} (\phi^q Z).
\end{equation}
This provides a different method of conformally transforming $\xi$ so that
\ifjournal
\goodbreak
$\int_M \tilde N \tilde \xi\; \widetilde{dV}=0$ 
\else
$\int_M \tilde N \tilde \xi\; \widetilde{dV}=0$ 
\fi
continues to hold.  For reasons we will return to and better 
motivate in future work, we call the vector field $Z$ a drift.

With drift-parameterized mean curvature, a conformal data set on a compact manifold $M^n$ is 
$(g,\sigma, \tau^*, Z, N)$, where $Z$ is an arbitrary vector field.  Given 
a conformal factor $\psi$, we conformally transform this data via
\begin{equation}
(\tilde g,\tilde \sigma, \tilde \tau^*, \tilde Z, \tilde N) = 
(\psi^{q-2} g,\psi^{-2} \sigma, \tau^*, Z, \psi^{q} \tilde N).
\end{equation}
In particular, $\tau^*$ and $Z$ are invariant under conformal transformations,
and we construct from them a conformally transformed mean curvature 
\begin{equation}
\tilde \tau = \tau^* + (1/\tilde N) \div_{\tilde g} Z.
\end{equation}
Using these transformation laws we seek a solution
of the constraint equations of the form
\begin{equation}
\begin{aligned}
\ol g &= \phi^{q-2} g\\
\ol K &= \phi^{-2} \left( \sigma + \frac{1}{2N}\ck_g W\right) + \frac{\tau^* + (\phi^{-2q}/ N)\div_{g} Z}{n} \ol g,
\end{aligned}
\end{equation}
and the constraint equations become
\ifjournal
\begin{equation}\label{eq:CTSH-DIV}
\begin{aligned}
-2\kappa q\Lap \phi + R\phi &=  \left| \sigma + \frac{1}{2N} \ck W\right|^2\phi^{-q-1} - \kappa(\tau^*+\phi^{-2q}\div(\phi^q Z))^2\phi^{q-1}
\\
-\div\left(\frac{1}{2N}\ck W\right) &= \kappa \phi^q \nabla \left( \frac{\phi^{-2q}}{N} \div(\phi^q Z) \right),
\end{aligned}
\end{equation}
\else
\begin{equation}\label{eq:CTSH-DIV}
\begin{aligned}
-2\kappa q\Lap \phi + R\phi - \left| \sigma + \frac{1}{2N} \ck W\right|^2\phi^{-q-1} + \kappa(\tau^*+\phi^{-2q}\div(\phi^q Z))^2\phi^{q-1}
&= 0\\
-\div\left(\frac{1}{2N}\ck W\right) &= \kappa \phi^q \nabla \left( \frac{\phi^{-2q}}{N} \div(\phi^q Z) \right),
\end{aligned}
\end{equation}
\fi
which we will call the \textbf{CTS-H}$^{**}$ equations.  The analytical difficulties
of the CTS-H$^{**}$ equations seem more substantial than those of the CTS-H$^*$ equations
because both the Hamiltonian and momentum constraints are second-order in $\psi$. 
However, we show here that as far as flat Kasner spacetimes are concerned,
the drift-parameterized mean curvature has properties that are essentially identical
to that of the lapse-scaled mean curvature of the CTS-H$^*$ equations.

Consider a CTS-H$^{**}$ conformal data set
\begin{equation}
(g_\Bell, \mu \sigma^\flat_\Bell, \tau^*, Z, N)
\end{equation}
on the torus $T^n$ where $\xi$ and $N$ are functions of $s=s^1$ alone,
and $Z = z \partial_{s}$ for some function $z(s)$.  As before,
we seek a solution of system \eqref{eq:CTSH-DIV} the form $\phi=\phi(s)$
and $W = w(s) \partial_s$. The CTS-H$^{**}$ equations become
\ifjournal
\begin{subequations}\label{eq:CTS-H-div-reduced}
\begin{alignat}{2}
-2\kappa q \ell^2\, \phi''  &= \kappa(\mu + (w'/N))^2\phi^{-q-1} - \kappa \left(\tau^*+\frac{\phi^{-2q}}{N}(\phi^q z)' \right)^2 \phi^{q-1}
\label{eq:CTS-H-div-red-ham} \\
(w'/N)' &= \phi^q \left(\frac{\phi^{-2q}}{N} (\phi^q z)'\right)'
. \label{eq:CTS-H-div-red-momentum}
\end{alignat}
\end{subequations}
\else
\begin{subequations}\label{eq:CTS-H-div-reduced}
\begin{alignat}{2}
-2\kappa q \ell^2\, \phi'' - \kappa(\mu + (w'/N))^2\phi^{-q-1} + \kappa \left(\tau^*+\frac{\phi^{-2q}}{N}(\phi^q z)' \right)^2 \phi^{q-1} &= 0
\label{eq:CTS-H-div-red-ham} \\
(w'/N)' &= \phi^q \left(\frac{\phi^{-2q}}{N} (\phi^q z)'\right)'
. \label{eq:CTS-H-div-red-momentum}
\end{alignat}
\end{subequations}
\fi
where $\ell=\Bell_1$ and derivatives are taken with respect to $s$. 

One readily verifies that 
if $\mu$ and $\tau^*$ are nonzero and have the same sign, then
\begin{equation}
\begin{aligned}
\phi \equiv c &= \left(\frac{\mu}{\tau^*}\right)^\frac{1}{q}\\
w &= z
\end{aligned}
\end{equation}
is a solution of equations \eqref{eq:CTS-H-div-reduced}, and that
if $\mu$ and $\tau^*$ both vanish then $\phi\equiv c$ and $w=z$ is a solution
for every $c>0$.  We then have the following straightforward variation of
Propositions \ref{prop:noncmc-star}, and we omit the proof.

\begin{proposition}\label{prop:noncmc-div}
Consider CTS-H$^{**}$ conformal data $(g_\Bell, \mu \sigma^\flat_\Bell, \tau^*, Z, N)$ on $T^n$
where $\mu$ is a constant, $N$ is a function of the first factor of $T^n$ only,
and $Z$ is a vector field of the form $Z=z(s^1)\partial_{s^1}$.

If $\mu$ and $\tau^*$ are nonzero and have the same sign, then
\begin{equation}
\begin{aligned}
\phi\equiv c &= \left(\frac{\mu}{\tau^*}\right)^\frac{1}{q}\\
W &= Z
\end{aligned}
\end{equation}
is a solution of the CTS-H$^{**}$ equations.
The solution of the constraint equations generated by this solution of the CTS-H$^{**}$
equations is
\begin{equation}
\begin{aligned}
\ol g &= g_{r\Bell} \\
\ol K &= \ol\tau\; (r\ell_1)^2 (ds^1)^2
\end{aligned}
\end{equation}
where 
\begin{equation}
\ol \tau = \tau^* + c^{-q} \frac{1}{N} \div_{g} Z
\end{equation}
and where
\begin{equation}
r=c^{(q-2)/2}=\left(\frac{\mu}{\tau^*}\right)^{\frac{1}{n}}.
\end{equation}

If $\mu$ and $\tau^*$ are both equal to zero, then for every
$c>0$ there is a solution $(\phi, W)$ of the CTS-H$^*$ equations such that
$\phi\equiv c$ everywhere and such that equation
\eqref{eq:ckW} holds.  The associated solution of the constraint equations is
given by equations \eqref{eq:nonCMCflat} with $r=c^{(q-2)/2}$.
\end{proposition}

It is easy to see that Proposition \ref{prop:class-star} also holds for the 
the CTS-H$^{**}$ equations with the statement trivially modified to refer to 
Proposition \ref{prop:noncmc-div} rather than Proposition \ref{prop:noncmc-star}.
Thus drift-parameterized mean curvature and lapse-scaled mean curvature
lead to similar parameterizations of the flat Kasner spacetimes:
we can predict in advance when we will obtain one-parameter families
of solutions, but we detect \textit{a-posteriori} if a solution of
the constraints is an expanding or contracting Kasner spacetime.

\section{Conclusion}

Our study of slices of Kasner spacetimes provides
a model for understanding
of the conformal method on a Yamabe-null manifold 
as we transition from CMC conformal data to the far-from-CMC regime. 

Starting with CMC slices of Kasner spacetimes, the conformal method yields
an effective parameterization.  The
Kasner exponents of the spacetime are determined in a predictable
way from the conformal data, and one-parameter families of static toroidal spacetimes
occur only in the degenerate case $\sigma\equiv 0$ and $\tau_0=0$.  In particular,
the flat Kasner spacetimes develop from conformal data sets of the form
\begin{equation}
(g_\Bell, \mu \sigma^\flat_\Bell, \tau\equiv\tau_0, N)
\end{equation}
where the constants $\mu$ and $\tau_0$ are nonzero and have the same sign.

Section \ref{sec:non-cmc} extended the flat Kasner analysis
to include all simple product slices.  
Within this category we saw that there is essentially no restriction
on the mean curvature. Given a conformal data set
\begin{equation}\label{eq:concdata}
(g_\Bell, \mu \sigma^\flat_\Bell, \tau, N)
\end{equation}
where $\tau$ and $N$ are functions of one factor of the torus,
if
\begin{equation}
\tau^* = \frac{\int_{T^n} N \tau\; dV_\Bell}{\int_{T^n} N\; dV_\Bell}
\end{equation}
is nonzero, and if $\mu$ has the same sign as $\tau^*$, then we pick up a unique
flat Kasner or static toroidal spacetime.   Trouble can only occur under the non-generic 
condition $\tau^*=0$.

If $\tau$ is positive everywhere, then $\tau^*>0$ as well, and there are no difficulties:
if the sign of $\mu$ matches that of $\tau^*$, then we pick up a single 
solution in the family of interest.  A similar analysis applies to negative mean curvatures,
and we see that the good properties of the CMC parameterization extend into
a near-CMC region that includes, regardless of the choice of $N$, the mean curvatures 
with constant nonzero sign.

Suppose, however, that $\tau$ is a mean curvature that changes sign. 
For most choices of $N$, $\tau^*\neq 0$ as well, and again there are no difficulties.
But for this particular $\tau$ there will be some choices of $N$ for which $\tau^*=0$, 
and in these cases we do not construct a flat Kasner or static toroidal slice unless $\mu=0$
as well, in which case we build a one-parameter family. Moreover, the condition
$\tau^*=0$ cannot be verified directly in terms of an arbitrary conformal data set that is conformally 
related to a data set of the form \eqref{eq:concdata}, and this poses a real difficulty.
We know that one-parameter families can appear, but we seem to lack a conformally covariant
way of detecting them \textit{a-priori}.

The challenge just described, in combination with the multiplicity and
non-existence phenomena found in \cite{Maxwell:2011if}, starts to paint a picture
that conformal method may not be well suited to the far-from-CMC setting.
One might have hoped that the phenomena discovered in \cite{Maxwell:2011if} 
were a consequence of the $L^\infty$ mean curvatures.
But the one-parameter families found here occur
even for smooth data; they are a genuine feature of the conformal method and they
provide potential counterexamples for theorems one might wish to prove.  For example,
any uniqueness theorem will have to have hypotheses that somehow rule out the data
sets found here that construct one-parameter families.

The negative findings of this study are tempered by the fact that
the conformal data sets considered here are very special:
the manifold is conformally flat and has conformal Killing fields, the mean curvatures are
highly symmetric, and so forth. One might still hope to prove positive results for the 
conformal method under some generic conditions.  Additionally, all of the difficulties found
here and in \cite{Maxwell:2011if} occur only for mean curvatures that change sign. So
the conformal method may yet provide a good parameterization of 
those solutions of the constraint equations that have mean curvatures with constant sign.
But there is now a growing body of evidence that conformal method encounters genuine difficulties
for far-from-CMC conformal data sets, and it seems reasonable to explore alternative constructions
such as those proposed in Section \ref{sec:mc-alt}.

\section*{Acknowledgment}
This work was supported by NSF grant 0932078 000 while I was a resident at 
the Mathematical Sciences Research Institute in Berkeley, California,
and was additionally supported by NSF grant 1263544.

\bibliographystyle{amsalpha-abbrv}
\bibliography{Kasners,Kasners-manual}

\end{document}